\documentclass[11pt]{amsart}
\usepackage{fullpage}
\usepackage[foot]{amsaddr}
\usepackage[OT1]{fontenc}
\usepackage{lmodern}
\usepackage{amsmath}
\usepackage{amssymb}
\usepackage{amsthm}
\usepackage[linesnumbered,ruled,vlined]{algorithm2e}
\DontPrintSemicolon
\usepackage{enumitem}
\usepackage{graphicx}
\usepackage{xcolor}
\usepackage[margin=1cm]{caption}
\usepackage{booktabs}

\usepackage[bookmarks=true,hypertexnames=false,pagebackref]{hyperref}
\hypersetup{colorlinks=true, citecolor=blue, linkcolor=red, urlcolor=blue}

\usepackage{cleveref}

\def\*#1{\mathbf{#1}}
\def\+#1{\mathcal{#1}}
\def\-#1{\mathrm{#1}}
\def\=#1{\mathbb{#1}}
\def\^#1{\mathbb{#1}}

\newcommand{\norm}[1]{\ensuremath{\left\lVert #1\right\rVert}}

\newcommand{\eps}{\varepsilon}

\newcommand{\diag}{\operatorname{diag}}

\newcommand{\stitle}[1]{\vspace{1ex} \noindent{\bf #1}}

\newtheorem{theorem}{Theorem}[section]
\newtheorem{lemma}[theorem]{Lemma}

\newtheorem{corollary}[theorem]{Corollary}

\newtheorem{condition}[theorem]{Condition}
\theoremstyle{definition}
\newtheorem{definition}[theorem]{Definition}

\theoremstyle{remark}

\crefname{theorem}{Theorem}{Theorems}
\crefname{observation}{Observation}{Observations}
\crefname{claim}{Claim}{Claims}
\crefname{condition}{Condition}{Conditions}
\crefname{algorithm}{Algorithm}{Algorithms}
\crefname{algocf}{Algorithm}{Algorithms}
\crefname{property}{Property}{Properties}
\crefname{example}{Example}{Examples}
\crefname{fact}{Fact}{Facts}
\crefname{lemma}{Lemma}{Lemmas}
\crefname{corollary}{Corollary}{Corollaries}
\crefname{definition}{Definition}{Definitions}
\crefname{remark}{Remark}{Remarks}
\crefname{proposition}{Proposition}{Propositions}
\crefname{equation}{equation}{equations}
\crefname{enumi}{Case}{Case}
\creflabelformat{enumi}{(#2#1#3)}

\definecolor{HGcolor}{RGB}{255,50,50}

\makeatletter
\def\prob#1#2#3{\goodbreak\begin{list}{}{\labelwidth\z@ \itemindent-\leftmargin
      \itemsep\z@  \topsep6\p@\@plus6\p@
      \let\makelabel\descriptionlabel}
  \item[\textbf{Name}]#1
  \item[\textbf{Instance}]#2
  \item[\textbf{Output}]#3
  \end{list}}
\makeatother

\newcommand{\one}{\mathbf{1}}

\newcommand{\Ot}{\widetilde{O}}

\newcommand{\vol}{\operatorname{vol}}

\begin{document}

\title{Nearly optimal Personalized PageRank computation}

\author{Rong-Hua Li}
\author{Quanqing Xu}
\author{Yichun Yang}
\author{Junjie Zhou}
\address[Rong-Hua Li]{School of Computer Science, Beijing Institute of Technology, Beijing, China}
\address[Quanqing Xu]{Independent Researcher, Beijing, China}
\address[Yichun Yang]{School of Computer Science, Beijing Institute of Technology, Beijing, China}
\address[Junjie Zhou]{School of Computer Science, Beijing Institute of Technology, Beijing, China}

\begin{abstract}
We study the fundamental problem of computing Personalized PageRank (PPR) on an undirected and unweighted graph $G$. We focus on the $\eps$-error guarantee introduced by Andersen, Chung, and Lang [ACL; FOCS 2006 $\&$ Internet Math 2007]. Their classic local push method computes an approximate PPR vector satisfying the ACL $\eps$-error guarantee in $O((\alpha\eps)^{-1})$ time, where $\alpha$ is the teleportation parameter. This problem has a lower bound $\Omega(\eps^{-1})$. Working with idealized SDDM solves, we eliminate the dependence on $\alpha$ and present an algorithm that achieves the same error guarantee in $O(\eps^{-1-o(1)})$ time. Consequently, we obtain a nearly optimal algorithm for PPR computation under the ACL error guarantee, and a nearly linear algorithm for local graph clustering with complexity independent of $\phi$.
\end{abstract}

\maketitle

\section{Introduction}\label{sec:intro}

As a fundamental problem in graph analysis, Personalized PageRank (PPR) has been extensively studied in the previous two decades. For a seed node $s$ and teleportation parameter $\alpha \in (0, 1)$, the PPR value from $s$ to a node $v$ can be viewed as the probability that a random walk that starts from $s$ terminates at $v$, where it terminates with probability $\alpha$ at each step.
Therefore, PPR provides a natural measure of node importance relative to a given seed set. It has been widely used in applications such as web search~\cite{PageEtAl1999,JehWidom2003}, recommendation systems~\cite{GoriPucci2007,ZhangEtAl2008,MustoEtAl2021}, and graph neural networks~\cite{KlicperaEtAl2019,BojchevskiEtAl2020,ChienEtAl2021}. The efficient computation of PPR has also attracted considerable attention in the database \cite{LofgrenEtAl2014,WangEtAl2017,WuEtAl2021,LiaoEtAl2022,LiaoEtAl2023} and theoretical computer science communities~\cite{ACL06,ACL07,WeiEtAl2024,JiangEtAl2026}.

Despite the considerable attention that PPR has received and its widespread use in applications, the complexity of its local computation remains poorly understood.
The study of local PPR computation dates back to Andersen, Chung, and Lang~\cite{ACL06,ACL07}, who proposed the classic Push algorithm.
This algorithm computes a PPR vector in $O((\alpha\eps)^{-1})$ time under the ACL $\eps$-error guarantee (defined later in \Cref{def:acl-error}).
One can easily see that a linear dependence on $\eps^{-1}$ is unavoidable under the ACL $\eps$-error guarantee.
Thus, the simple Push algorithm is already optimal when $\alpha$ is a fixed constant. However, when $\alpha$ is small, as is typical in some applications such as local graph clustering~\cite{ACL06,ACL07}, the dependence on $\alpha$ becomes significant.

Along this direction, Fountoulakis et al.~\cite{FRSCM19} introduced $\ell_1$-regularized PageRank as an equivalent formulation of the problem under the ACL error guarantee.
Subsequently, at COLT 2022, Fountoulakis and Yang~\cite{FY22} posed the open problem of whether an $O((\sqrt{\alpha}\,\eps)^{-1})$ algorithm can be designed for $\ell_1$-regularized PageRank.
Equivalently, they asked whether an algorithm with this running time exists under the ACL $\eps$-error guarantee.
Later, Mart\'inez-Rubio, Wirth, and Pokutta~\cite{MWP23} obtained an $O(1/(\sqrt{\alpha}\,\eps^2))$ algorithm and, somewhat surprisingly, also designed an $O(\eps^{-3})$ algorithm whose complexity is independent of $\alpha$.
More recently, using a novel active-set method, Wei and Yang~\cite{WY26} proposed an $O\bigl(\eps^{-2}\,\mathrm{polylog}((\alpha\eps)^{-1})\bigr)$ algorithm, improving upon the result of Mart\'inez-Rubio et al.~\cite{MWP23}.
Several follow-up works further advanced this line of research:
Cui, Wei, and Yang~\cite{CWY26} improved the complexity to $\Ot\bigl(\min\{\eps^{-2},\, 1/(\sqrt{\alpha}\,\eps)\}\bigr)$;
concurrently, Zhou~\cite{Zhou26} obtained a deterministic version of this result;
and Li and Yang~\cite{LY26} extended similar ideas to $\ell_1$-regularized resistance and developed faster algorithms for counting spanning trees.
We refer to \Cref{tab:ppr-complexity} for a summary of these recent results.
Despite this substantial progress, the complexity landscape of PPR computation remains far from clear.
This naturally raises the following question: \textit{what is the optimal dependence of PPR algorithms on $\eps$ and $\alpha$?}

\begin{table}[h]
\centering
\caption{Local algorithms for PPR under ACL $\eps$-error guarantee.}
\label{tab:ppr-complexity}
\renewcommand{\arraystretch}{1.2}
\begin{tabular}{ll}
\toprule
Method & Complexity \\
\midrule
\cite{ACL06,ACL07} & $O((\alpha\eps)^{-1})$ \\
\cite{FRSCM19} & $\Ot((\alpha\eps)^{-1})$ \\
\cite{MWP23} & $O(1/(\sqrt{\alpha}\,\eps^{2}))$ \\
\cite{MWP23} & $O(\eps^{-3})$ \\
\cite{WY26} & $O\bigl(\eps^{-2}\,\mathrm{polylog}((\alpha\eps)^{-1})\bigr)$ \\
\cite{CWY26} & $O\bigl(\min\{\eps^{-2},\, 1/(\sqrt{\alpha}\,\eps)\}\mathrm{polylog}((\alpha\eps)^{-1})\bigr)$ \\
\cite{Zhou26} & $O\bigl(\min\{\eps^{-2},\, 1/(\sqrt{\alpha}\,\eps)\}\mathrm{polylog}((\alpha\eps)^{-1})\bigr)$ \\
\midrule
\textbf{This work} & $O(\eps^{-1-o(1)})$ \\
\textbf{Lower bound} & $\Omega(\eps^{-1})$ \\
\bottomrule
\end{tabular}
\end{table}

We first formally state the ACL $\eps$-error guarantee considered in this paper.

\begin{definition}[ACL $\eps$-error guarantee]\label{def:acl-error}
Let $\pi_s$ denote the PPR vector with source node $s$, and let $d_u$ denote the degree of node $u$.
An approximate PPR vector $\hat{\pi}_s$ satisfies the ACL $\eps$-error guarantee if
\[
  |\hat{\pi}_s(u)-\pi_s(u)| \leq \eps d_u
  \qquad \text{for every } u\in V.
\]
\end{definition}

This problem has a clear lower bound. For any fixed \(\alpha\), a star with its center as the source has \(\Theta(1/\varepsilon)\) leaves satisfying \(\pi_s(v)>\varepsilon d_v\), so the ACL guarantee forces that many nonzero estimates, yielding an \(\Omega(1/\varepsilon)\) time lower bound for explicit output.

Under the SDDM convention specified in \Cref{sec:preliminaries}, our main result gives a nearly optimal dependence on $\eps$ while removing the dependence on the teleportation parameter $\alpha$.

\begin{theorem}[Main result]\label{thm:main}
There is an algorithm that computes an approximate PPR vector satisfying the ACL $\eps$-error guarantee in
\(
  O\bigl(\eps^{-1-o(1)}\bigr)
\)
expected work in this model. In particular, the bound is independent of $\alpha$ and  nearly optimal to $\Omega(\eps^{-1})$ lower bound.
\end{theorem}

As a direct consequence, to the best of our knowledge, our PPR algorithm yields the first nearly linear time algorithm for local graph clustering with complexity independent of $\phi$.

\begin{corollary}\label{cor:local-clustering}
For a reference set $C$ of volumn $\vol(C)=k$ conductance $\phi$, there is a local algorithm that finds a set $S$ of conductance $O(\sqrt{\phi\log k})$ and volume $O(k)$ in \(k^{1+o(1)}\) time.
\end{corollary}

Indeed, the local clustering framework of Andersen, Chung, and Lang~\cite{ACL06,ACL07} uses the Push algorithm to compute an ACL $\eps$-approximate PPR vector with support volume $O(1/\eps)$. They set $\eps=\Theta(1/k)$ and $\alpha=\Theta (\phi)$ to find the cut $S$.
Substituting the algorithm from \Cref{thm:main} into this framework gives a PPR computation time of $\eps^{-1-o(1)}=k^{1+o(1)}$. Thus, the resulting local clustering algorithm runs in nearly linear time in the target volume $k$. This improves upon the $O(k/\phi+k\log k)$ runtime by \cite{ACL06,ACL07} and the $O(k^2\mathrm{polylog} \frac{k\log n}{\phi})$ runtime by \cite{WY26}.

Beyond its application to local graph clustering, we believe that our result has the potential to improve a range of other problems, including effective resistance computation~\cite{YangEtAl2025ER,CuiEtAl2025ER}, sublinear-time solvers for diagonally dominant linear systems~\cite{KwokEtAl2026}, and computation in the Friedkin--Johnsen model~\cite{FengEtAl2026}. Establishing these consequences would require additional problem-specific reductions and analyses. Since these directions are beyond the scope of this paper, we leave them to future work.

\begin{figure}[!ht]
 \centering
 \includegraphics{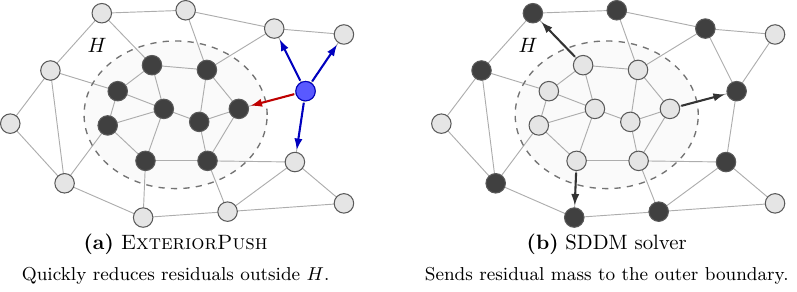}
 \caption{Residual distributions after local push and SDDM solve on $H$.
 Darker vertices indicate larger residuals.}
 \label{fig:push-sddm-residuals}
\end{figure}

\subsection{Overview of our method}\label{sec:overview}

Our method combines two existing algorithms for local PPR computation. The first is the Push algorithm of Andersen, Chung, and Lang~\cite{ACL06,ACL07}, which pushes residual mass from a single vertex at a time and can be viewed as a coordinate-wise Gauss--Seidel iteration. Although each push is local, the same vertex may need to be processed repeatedly. This leads to an $O(1/(\alpha\eps))$ running time, with a linear dependence on $1/\alpha$.

The second is the SDDM solver of Spielman and Teng~\cite{ST14}, subsequently used by Wei and Yang~\cite{WY26} to design the active-set method for PPR computation. Given an explored active set $H$ with $\vol(H)=O(1/\eps)$, an SDDM solve takes time nearly linear in $\vol(H)$ and depends only logarithmically on the residual accuracy. However, in the active-set framework of Wei and Yang, each iteration may discover only one new vertex outside $H$. Thus, as many as $O(1/\eps)$ iterations may be needed, with each iteration solving a system on the growing active set. This repeated work leads to a quadratic dependence on $1/\eps$, giving a total running time of $O\bigl(\eps^{-2}\,\mathrm{polylog}((\alpha\eps)^{-1})\bigr)$.

To avoid this repeated work, we combine Push and SDDM solves in a new framework. We design \textnormal{\textsc{ExteriorPush}}, a variant of ACL Push that quickly discovers a active set $H\subseteq V$ in which large residuals are concentrated inside $H$, while ensuring that the residuals outside $H$ are small, see \Cref{fig:push-sddm-residuals}(a). We then apply the SDDM solver inside $H$ to make its internal residuals sufficiently small, sending residual mass across its boundary to vertices outside $H$, see \Cref{fig:push-sddm-residuals}(b). The two algorithms therefore complement each other: \textnormal{\textsc{ExteriorPush}} efficiently explores the graph and controls residuals outside the active set, while the SDDM solver reduces residuals inside it through a joint update.

However, the main difficulty is to coordinate these two operations within a single framework while controlling the total work of their repeated calls. Carrying out this framework is more involved. We first present in \Cref{sec:three-halves} a simpler algorithm that runs in $\widetilde O(\eps^{-3/2})$ time and already breaks the quadratic bound of \cite{WY26}. Its basic idea is to use \textnormal{\textsc{ExteriorPush}} to discover $H$ and control the residuals outside $H$, and then to apply SDDM solves repeatedly to the high-residual vertices inside $H$ to reduce those residuals efficiently. This framework is subsequently used to design a recursive algorithm (\Cref{sec:recursive-reduction}), which recurses on subgraphs of successively smaller volume and finally yields the $O(\eps^{-1-o(1)})$ expected running time in \Cref{thm:main}, independent of $\alpha$.

\stitle{AI Declaration.} The result of this paper is derived by the interaction between the authors and the GPT-6-Astra. We have independently checked and rewritten all the contexts.

\section{Preliminaries}\label{sec:preliminaries}

Let $G=(V,E)$ be a simple unweighted and undirected graph, and let $s\in V$ be the seed node. We write $A\in\mathbb R^{V\times V}$ for the adjacency matrix, $d_u$ for the degree of a node $u$, $\mathbf d=(d_u)_{u\in V}$ for the degree vector, and $D:=\diag(\mathbf d)$ for the degree matrix. The graph Laplacian is $L:=D-A$, the probability transition matrix is $P=AD^{-1}$.
For a teleportation parameter $\alpha\in(0,1)$, the personalized PageRank vector $\pi_s\in\mathbb R^V$ is the unique solution to
\[
  \pi_s=\alpha\mathbf e_s+(1-\alpha)P\pi_s,
\]
where $\mathbf e_s$ is the one-hot vector corresponding to $s$. Equivalently, defining $M:=D-(1-\alpha)A$, we can write $\pi_s=\alpha DM^{-1}\mathbf e_s$.

We count arithmetic operations and adjacency-list queries.  Assume
$0<\eps\le1$ and a nonisolated source; an isolated source is returned
exactly.  Matrix identities are understood on the positive-degree
vertices, with zero output at other isolated vertices.  

An SDDM matrix is a symmetric positive-definite, diagonally dominant matrix whose off-diagonal entries are nonpositive. The matrix $M$ is an SDDM matrix, and every principal submatrix $M_{SS}$ is also SDDM with an entrywise nonnegative inverse. Throughout the algorithm, we use the following nearly linear time SDDM solver as a key subroutine. For the algorithmic analysis, an SDDM solve on a set
$T$ is an idealized primitive: it returns the exact solution and is
charged $\widetilde O(\vol(T))$ work. 

\begin{lemma}[Nearly linear time SDDM solver~\cite{ST14}]\label{lem:sdd}
  Let $M\in\mathbb R^{n\times n}$ be an SDDM matrix with $q$ nonzero entries, let $b\in\mathbb R^n$, and let $0<\eta<1$. One can compute $\widetilde x$ satisfying
  \[
   \norm{\widetilde x-M^{-1}b}_M
   \le \eta\norm{M^{-1}b}_M\footnotemark
  \]
  in $\Ot(q\log(1/\eta))$ time with probability at least $1-n^{-O(1)}$.
  \footnotetext{For any vector $z\in\mathbb R^n$, we define $\norm{z}_M:=\sqrt{z^{\mathsf T}Mz}$.}
\end{lemma}

For any vector $z\in\mathbb R^V$, define its residual by
$\mathbf e_s-Mz$.  The residual lets us check the PageRank error
without knowing the exact PageRank vector.

\begin{lemma}[Residual and approximation error]
\label{lem:acl-residual}
For any $z\ge0$, the approximation $\widehat\pi_s=\alpha Dz$ satisfies
\[
 \pi_s-\widehat\pi_s=\alpha DM^{-1}(\mathbf e_s-Mz).
\]
In particular, $0\le \mathbf e_s-Mz \le\eps\mathbf d$ implies that
$\widehat\pi_s$ satisfies the ACL $\eps$-error guarantee.
\end{lemma}

\begin{proof}
The identity follows from $\pi_s=\alpha DM^{-1}\mathbf e_s$.
Since $M^{-1}$ is entrywise nonnegative and
$M\one=\alpha\mathbf d$, the residual bound gives
\[
 0\le\pi_s-\widehat\pi_s
   =\alpha DM^{-1}(\mathbf e_s-Mz)
   \le\alpha DM^{-1}(\eps\mathbf d)
   =\eps\mathbf d.
\]
\end{proof}

\stitle{The \textnormal{\textsc{Push}} algorithm.}
We recall the classic \textnormal{\textsc{Push}} algorithm of \cite{ACL06,ACL07}. It maintains a vector $z\ge0$ and its residual
$r:=\mathbf e_s-Mz$. By \Cref{lem:acl-residual}, the goal is to reach
$0\le r\le\eps\mathbf d$. A \emph{push} on a vertex $v$ increases $z_v$ by
$t=r_v/d_v$. By the definition of residual,
we reset $r_v$ to zero and adds $(1-\alpha)t$ to $r_u$ for each neighbor
$u$ of $v$, leaving all other coordinates unchanged. A single push operation at $v$ takes $O(d_v)$ queries. The \textnormal{\textsc{Push}} algorithm
starts from $z=0$ and $r=\mathbf e_s$, and pushes a vertex $v$ whenever
$r_v>\eps d_v$, until no such vertex remains; see \Cref{alg:push}.

\begin{algorithm}[H]
\caption{\textnormal{\textsc{Push}}~\cite{ACL06,ACL07}}
\label{alg:push}
\small
\KwIn{$G,s,\alpha\in(0,1)$ and $\eps\in(0,1]$.}
\KwOut{An ACL $\eps$-approximation $\widehat\pi_s$.}
$z\leftarrow0$; $r\leftarrow\mathbf e_s$\;
\While{there is a vertex $v$ with $r_v>\eps d_v$}{
 $t\leftarrow r_v/d_v$\;
 $z_v\leftarrow z_v+t$; $r_v\leftarrow0$\;
 \lFor{$u\in N(v)$}{$r_u\leftarrow r_u+(1-\alpha)t$}
}
\Return $\widehat\pi_s=\alpha Dz$\;
\end{algorithm}

\begin{lemma}[\cite{ACL06,ACL07}]\label{lem:push}
\Cref{alg:push} runs in $O(1/(\alpha\eps))$ time and returns $z\ge0$ with
\[
 0\le\mathbf e_s-Mz\le\eps\mathbf d.
\]
Consequently, $\widehat\pi_s=\alpha Dz$ satisfies the ACL $\eps$-error
guarantee.
\end{lemma}

Intuitively, the factor $1/\alpha$ arises because residual sent to a neighbor may later
return, so the same vertex can be pushed many times, while each push settles
only an $\alpha$ fraction of the residual. We refer to \Cref{sec:three-halves} for the detailed analysis.

\section{The \texorpdfstring{$\widetilde O(\eps^{-3/2})$}{three-halves} non-recursive algorithm}
\label{sec:three-halves}

In this section, we present an algorithm that breaks the previous $\Ot(\eps^{-2})$ barrier \cite{WY26} (see \Cref{nr:runtime}). Its design and analysis are not directly used in the proof of \Cref{thm:main}. Nevertheless, it provides a framework whose core ideas and proof strategy will be reused in the more involved algorithms that follow. The design and analysis of the nearly linear time algorithm are given in \Cref{sec:recursive-reduction,sec:technical-proofs}.

Following the strategy of $\ell_1$-regularized PageRank \cite{MWP23,WY26}, we discover an growing active set $H$ such that directly invoking the SDDM solver (\Cref{lem:sdd}) on the final $H$ is enough for the ACL guarantee. The central task is to discover $H$ such that: (i) $\vol(H)=O(1/\eps)$; (ii) the time to discover $H$ is cheap and (iii) after the SDDM solve on $H$, the residual satisfies $0\le \mathbf e_s-Mx \le\eps\mathbf d$. Then we use \Cref{lem:acl-residual} to show that the ACL guarantee is satisfied. 

Condition (i) can be easily satisfied: we retain $\Theta(\eps d_v)$ residual at each of the vertex of $H$ (i.e., $\mathbf (e_s-Mx)_v=\Theta(\eps d_v)$ for $v\in H$). Since $\|\mathbf e_s-Mx\|_1\le 1$, this would give
$\vol(H)=O(1/\eps)$, so the final SDDM solve would cost only
$\widetilde O(1/\eps)$. However, the condition (iii) is essential:
solving on $H$ must also leave the exterior $V\setminus H$ within the residual
bound. An arbitrary small region need not have this property. 

The previous paper \cite{WY26} use $\Ot(\eps^{-2})$ time to discover the active set $H$ satisfies (i), (iii) by repeated SDDM solves on $H$ and discover large residuals on $V\setminus H$. To break the quadratic bound and make the process more efficient, we combine the idea of Push algorithm \cite{ACL06,ACL07} and the SDDM solver. There are two costs to control. Ordinary push can discover new vertices,
but repeated visits introduce a dependence on $1/\alpha$. Repeated SDDM
solves avoid those intermediate visits, but may cost quadratically in
$1/\eps$ if each solve reveals very little new vertices in $V\setminus H$. We first make discovery
inexpensive, then determine when a whole-region solve is worth its cost.

Write $q=\eps/96$ and $N=1/q$. The parameter $b$
is the volume scale of the smaller regions we will solve on. We discuss
$0<\alpha\le1/2$; ordinary push already gives $O(1/\eps)$ work for
$\alpha>1/2$.

\begin{theorem}[Nonrecursive algorithm]\label{nr:runtime}
For any integer $1\le b\le N$, \Cref{nr:alg-main} computes an ACL
$\eps$-approximation in $\widetilde O(N^2/b+Nb)$ expected time, using
$O(N)$ graph adjacency queries. In particular, $b\asymp\sqrt N$
gives $\widetilde O(\eps^{-3/2})$ expected time.
\end{theorem}


\subsection{Making discovery inexpensive}
\label{nr:design}

We start with \textnormal{\textsc{Push}} because it reveals the graph one adjacency list at
a time. A push that increases $x_v$ by $t$ costs $O(d_v)$, but decreases
the total mass of the residual $\|\mathbf e_s-Mx\|_1$ by only $\alpha d_vt$ (\Cref{alg:push}). Thus even a push
from residual above $qd_v$ may remove only $\Theta(\alpha qd_v)$ mass.
To avoid paying $1/\alpha$ for repeated visits, we separate the pushes
into two kinds that can be paid for differently. Note that all our push operations work outside $H$.

\begin{itemize}[leftmargin=*,itemsep=3pt,topsep=4pt]
\item \emph{Pushes that admits a new vertex to $H$ (admitting push).}
Once a vertex joins $H$, hand it over to SDDM solves and never push it
individually again. There is at most one admitting push per active
vertex, so these scans cost $O(\vol(H))$ in total. By bounding $\vol(H)\le N$, this is $O(N)$.
\item \emph{Pushes that do not admit a new vertex (nonadmitting push).}
These may revisit a vertex. Each must instead remove
$\Omega(qd_v)$ from the sum of work residual (which we discribe later), so
that the removed mass pays for its $O(d_v)$ scan.
\end{itemize}

The second requirement motivates a change in what we push. Outside
$H$, we temporarily set aside residual proportional to $d_vx_v$ and
process only the remainder, which we call the \emph{work residual}.
Increasing $x_v$ now both propagates residual and increases this reserve.
The work residual therefore decreases substantially even when $\alpha$
is small. We cap the potential accumulated in this way: once $x_v$
reaches an activation threshold, $v$ joins $H$ and subsequent updates
are handled by SDDM solves. This gives exactly the two kinds of push
above.

Concretely, let $c=3q/8$. An exterior vertex retains $4d_vx_v$ and
has $x_v<c$; an active vertex retains $qd_v$ and has $x_v\ge c$.
The maintained nonnegative state satisfies
\begin{align}
 \mathbf e_s-Mx=q\mathbf d_H+4D_{V\setminus H}x+r,
 \label{nr:state}
\end{align}
where $\mathbf d_H$ equals $\mathbf d$ on $H$ and vanishes elsewhere, $r$ is the work residual.
We push an exterior vertex only when $r_v>qd_v$, using
$t=\min\{r_v/(5d_v),c-x_v\}$. The update on $v$ consumes $d_vt$
and the reserve consumes another $4d_vt$; each neighbor receives
$(1-\alpha) t$. The minimum stops the push either when its work residual
is cleared or when the vertex is ready to join $H$. At activation,
replacing the reserve $4cd_v$ by $qd_v$ releases $qd_v/2$ back into
work residual. These are the updates in \Cref{nr:alg-exterior}.

We call this push procedure \textnormal{\textsc{ExteriorPush}}, and
call running it until every vertex $v\in V\setminus H$ satisfies $r_v\le qd_v$ a
\textit{drain}. The new \textnormal{\textsc{ExteriorPush}} guarantees two
properties: (i)  \textnormal{\textsc{ExteriorPush}} is a “free lunch”: every drain in the entire algorithm shares one $O(N)=O(\eps^{-1})$ total discovery cost (\Cref{nr:push-account}); and (ii) after a drain, the residual outside $H$ fits within the approximation error: $(\mathbf e_s-Mx)_v\le\eps d_v$ for every $v\in V\setminus H$.

For (i), intuitively, the total residual mass satisfies
$\|\mathbf e_s-Mx\|_1\le1$, and \textnormal{\textsc{ExteriorPush}}
performs only the two kinds of push described above. An admitting push
happens at most once per vertex and costs $O(d_v)$. After admission,
$v$ retains $qd_v$ residual, so $q\vol(H)\le\|\mathbf e_s-Mx\|_1\le1$
and $\vol(H)\le N$. Hence all admitting pushes cost
$O(\vol(H))=O(N)$. A nonadmitting push also costs $O(d_v)$.
It decreases the total work residual by $(4+\alpha)d_vt=\Omega(qd_v)$.
The total work residual is at most $\|\mathbf e_s-Mx\|_1\le1$, plus the newly activated work residual $q\vol(H)/2\le 1/2$, hence the total cost of nonadmitting pushes is also $O(N)$.
For (ii), the guarantee follows directly from the stopping condition of
a drain. For $v\notin H$, we have $r_v\le qd_v$ and $x_v<c$, so
\eqref{nr:state} gives
$(\mathbf e_s-Mx)_v=4d_vx_v+r_v\le(4c+q)d_v=\tfrac52qd_v\le\eps d_v$.

So next our goal is to make the residual inside $H$ also small. This suggests an \emph{intuitive algorithm}: start at $s$ and drain;
SDDM solve on the current $H$; send the residual to $V\setminus H$;
and repeat until $0\le\mathbf e_s-Mx\le\eps\mathbf d$. The SDDM solve leaves
very little work residual within $H$, but its boundary flow may
activate more vertices. All discovery in this algorithm is $O(N)$ by \textnormal{\textsc{ExteriorPush}}. Its unresolved cost is the repeated SDDM solve: a drain
may change only a small part of $H$, yet the next solve processes all
of it. We now ask what would justify paying for that solve.

\subsection{Making a whole-region solve worthwhile}
\label{nr:full-idea}

A path illustrates the difficulty with the intuitive algorithm. Place
the source $s$ at middle and consider a round in which the drain of \textnormal{\textsc{ExteriorPush}}
adds only the next vertex at each end of the active interval. Repeating
such rounds would mean solving on a growing interval to discover just
two more vertices each time. As a result, the total complexity falls back to $O(\eps^{-2})$.

However, we notice that the large residuals are just concentrated near the two newly admitted endpoints.
This suggests designing an algorithm invoking SDDM solver just \emph{locally} around those high residual vertices.
The question is how such local work can guarantee better complexity.

Before answering this question, we first analyze the intuitive algorithm. Fix three moments: $t_0$, just before a full SDDM solve on $H$ and after a drain; $t_1$, immediately after that SDDM solve; and $t_2$,
the next drained state at which another full SDDM solve on new $H$ is about to begin.
Write $H=A$ for the active set at $t_0$ and $H=A\cup B$ for the active set
at $t_2$. Let $x^0,r^0$ be the state at $t_0$. The first SDDM solve increases
$x_A^0\leftarrow x_A^0+h$ and leaves work residual $f$ on $A$ at $t_1$ ($f$ can  be very small). Now we upper bound the total work residual mass $\mu_{\rm next}$ at $t_2$.

At $t_1$, the old region $A$ therefore contributes only the small mass
$\|f\|_1$. The potentially large contribution to $\mu_{\rm next}$ is on the newly admitted vertices $B$, which receive boundary flow from $h$. Split that increment $h$ into two parts: $h_v=\min\{h_v,x_v^0\}+w_v$, where
$w_v=(h_v-x_v^0)_+$. The first term sends no more flow to any
$u\in B$ than the old potential $x_A^0$ contributes to the residual
balance at $u$. Since $u$ was exterior and drained at $t_0$, we have
$x_u^0<c$ and $r_u^0\le qd_u$. Combining this with residual identity \eqref{nr:state} bounds the incoming flow
attributable to the old potential $x_{A}^0$ by $O(qd_u)$. Summing over $u\in B$ gives $O(q\vol(B))$.

For the excess $w$, we bound the flow where it originates. Its support
is $D=\{v\in A:h_v>x_v^0\}$ (i.e., $x_v$ \textit{doubles} in the first full SDDM solve). The excess flow reaching $B$ is at most the total
that leaves $D$. To bound this outflow leaving $D$, let $x^1$ be the potential at
$t_1$ and note that $w_v=x_v^1-2x_v^0$ on $D$. At the comparison
potential $2x^0$, the residual at $v\in D$ is at most
$2(qd_v+r_v^0)$: the flow terms double, while the source input stays
fixed. At $x^1$, the residual is $qd_v+f_v$. The difference is thus
at most $qd_v+2r_v^0$. On summing over $D$, this bounds the excess outflow by $q\vol(D)+2\sum_{v\in D}r_v^0$.

The initial work residual and the change of reserve on $B$ contribute
another $O(q\vol(B))$. Between $t_1$ and $t_2$, any further inflow
into $A\cup B$ from outside is bounded by the total outflow of
\textnormal{\textsc{ExteriorPush}}. Combining these terms with the
remainder $\|f\|_1$ on $A$ gives the following comparison lemma.

\begin{lemma}[Residual after a SDDM solve]
\label{nr:progress}
Let $\mu$ and $\mu_{\rm next}$ be the active work-residual masses
at $t_0$ and $t_2$, respectively. Set
$J=12q\vol(B)+U$, where $U$ is the total outgoing residual flow of \textnormal{\textsc{ExteriorPush}} between the two states. Then
\begin{align}
 \mu_{\rm next}\le J+q\vol(D)+2\sum_{v\in D} r_v^0+\|f\|_1.
 \label{nr:progress-bound}
\end{align}
\end{lemma}

The proof in \Cref{nr:analysis} uses only the state identities and
the drained endpoints. Thus the comparison is a property of the
intuitive framework itself; it remains available after we add local
solves. We use it to divide transitions between full solves into two
cases.

\begin{itemize}[leftmargin=*,itemsep=4pt,topsep=4pt]
\item \emph{Good case: the active work-residual mass halves ($\mu_{\rm next}\le\mu/2$).}
If $\mu_{\rm next}\le\mu/2$, successive such transitions reduce
the mass geometrically toward the final accuracy scale. We count
consecutive runs of good cases: a bad case can increase the mass
again, but each new run is still logarithmically long.
\item \emph{Bad case: the mass does not halve ($\mu_{\rm next}>\mu/2$).}
If $\mu_{\rm next}>\mu/2$, we consider two cases. First suppose the new input is large, say $J\ge\mu/4$. Notice that newly active
volume can be charged only once, and the total exterior-push flow is
bounded. Thus $J$ has a constant global sum. Above the stopping scale, say
$\mu=\Omega(qb)$, each such event spends $J=\Omega(qb)$ of that sum. Intuitively, this implies $\vol(B)=\Omega(b)$ (when ignoring $U$), which can not happen so many times since $\sum \vol(B)\le \vol(H)\le N$.
Otherwise $J<\mu/4$. With the full SDDM solve residual $f$ chosen small enough,
\Cref{nr:progress} forces $q\vol(D)+2\sum_{v\in D} r_v^0=\Omega(\mu)$.
This is the \textit{difficult case} that full solves must make more efficient.
\end{itemize}

The obstruction is now explicit. We hope each difficult case imply substantial
doubling volume of $D$ with the condition $q\vol(D)+2\sum_{v\in D} r_v^0=\Omega(\mu)$. However, a small set $D$ might carry a large
fraction of the residual. We need to control $2\sum_{v\in D} r_v^0$ through $\vol(D)$.
It suffices to bound every active residual density before the first
full solve. This motivates the following intermediate task.

\noindent\textit{Output goal (fixed cache volume scale $b$ and density threshold
$\theta\ge q$).} Starting from a drained state, reach another drained
state in which $r_v<3\theta d_v$ for every $v\in H$, using joint
updates only on caches of volume less than $2b$. The target includes
vertices that join $H$ during this computation.

Here \textit{density} represents for \textit{residual mass} divided by \textit{volume}. Suppose this goal is achieved with density $\theta=O(\mu/b)$ at the next
full-solve request. Then $2\sum_{v\in D} r_v^0=O(\mu\vol(D)/b)$ and $q\le\theta = O(\mu/b)$.
The difficult case therefore requires $\vol(D)=\Omega(b)$.
Each vertex enters $H$ with potential $c$, only increases thereafter,
and, until component completion, has potential at most $3N$
(\Cref{lem:potential-completion}). It can thus belong to a doubling
set only $O(\log N)$ times. So large doubling sets $D$ cannot occur often.
The large-input account and this doubling account bound bad cases;
the good cases between them are bounded by geometric decrease.
\Cref{nr:halving,nr:full-exceptions,nr:counts} make this argument quantitative.

We will obtain the output goal by \emph{cache solves}, and request
a full SDDM solve only after establishing the required density bound,
apart from the initial solve and the restarts accounted for below.
This also identifies the eventual recursive algorithm in \Cref{sec:recursive-reduction}: reducing residuals
on a small cache is a smaller instance of the same task as reducing
them on the whole active set.

\subsection{Controlling residual with a small cache}
\label{nr:cache-idea}

We now implement the output goal with cache scale $b$ and target
density $3\theta$. The idea is to collect the active vertices with
excessive work residual in a small cache set $C$ and solve only on that set.
We repeat these local solves until every vertex in the entire active
set meets the density target. If the cache stays small throughout,
we obtain the condition that makes the difficult case imply
large doubling volume, without first paying for another whole-region
solve.

Start with an empty cache. After each exterior drain, admit every
active vertex with $r_v\ge3\theta d_v$. A
\textnormal{\textsc{CacheSolve}} reduces the cached residuals $\theta d_v\le r_v\le (\theta+q/4)d_v$; its full procedure
is \Cref{nr:alg-cache}. Its boundary flow may raise residuals at other
active vertices or cause exterior discovery. We therefore run \textnormal{\textsc{ExteriorPush}} and drain again,
admit any newly excessive active vertices, and repeat. This is the
loop in lines~\ref{nr:line-normalize-start}--\ref{nr:line-cache-call}
of \Cref{nr:alg-main}.

Two features let us control the size of all caches used during this
loop. The cache grows monotonically, so a vertex already processed
remains available when residual returns to it. Each vertex is admitted
only once during this attempt. Also, every cache solve retains the
positive floor $\theta d_v$. Previously admitted vertices therefore
continue to contribute at least this much work residual. At every
later solve, $\theta\vol(C)\le \sum_{v\in C} r_v$, even if the rows that triggered
the current solve are different from those that triggered earlier
ones. Controlling the total active mass would consequently control
the accumulated cache volume. The remaining difficulty is that \textit{this mass may
grow as exterior pushes send input inward and new vertices join $H$}.

We handle that mass growth caused by exterior pushes and new active vertices with a \emph{budget}. First choose the scale for
one attempt, called an \emph{episode}. At its drained starting state,
let the active mass be $\mu_0$, and fix
$Q=\max\{\mu_0,qb\}$ and $\theta=Q/b$ for the whole episode.
Thus the starting mass is at most $\theta b$. The lower bound on $Q$
ensures $\theta\ge q$; reaching the output goal at $\theta=q=\Theta(\epsilon)$
will give the final approximation. The episode allows an additional
amount $Q$ of charged input, recorded by a budget counter $W$ initially zero.
If a full solve has been requested, it occurs after this snapshot and
before the cache loop. Its boundary flow is drained with the budget
already open.

The counter $W$ has two contributions. When every new vertex $v$ joins $H$, we add
$3qd_v/2$: this covers up to $qd_v$ of work residual at the drained
starting snapshot and another $qd_v/2$ for the reserve change.
Its further growth is
accounted for by the updates between them. For every exterior push
of size $t$ at $v$, we also add its entire outgoing flow $(1-\alpha) d_vt$.
Counting the whole outflow bounds whatever part enters the eventual
active set. In particular, if $H^{\rm start}$ is the starting set and
$U$ is this accumulated exterior outflow by \textnormal{\textsc{ExteriorPush}}, then
\begin{align}
 W=\frac{3q}{2}\vol(H\setminus H^{\rm start})+U.
 \label{nr:charge}
\end{align}
Joint solves receive no additional charge. When residual is summed
over the final active set, their internal transfers cancel and their
remaining contribution is nonpositive. Thus this budget controls
the growth of active mass without charging the same mass each time
a joint solve redistributes it. The two updates to $W$ appear
explicitly in \Cref{nr:alg-exterior}.

We stop the episode immediately after $W$ reaches $Q$, before
performing another cache solve. This rule gives the required size
guarantee.

\begin{lemma}[The budget keeps the cache small]\label{nr:endpoint}
At every endpoint of an episode, its active work-residual mass satisfies
$\mu\le\mu_0+W$. While $W<Q$, every permitted cache solve has
$\vol(C)<2b$.
\end{lemma}

Indeed, an unexhausted budget gives $\mu<2Q$, while the retained
floor gives $\theta\vol(C)\le \sum_{v\in C} r_v\le\mu$. Since $Q=\theta b$,
the cache volume is below $2b$. The operation that exhausts the budget may
break the size bound, but no further cache solve is allowed in that
episode.

An exhausted budget signals enough new input to require a fresh
attempt. We close the episode, finish draining, and restart with a
full solve and a newly chosen allowance. These restarts are affordable:
each spends at least $qb$, whereas the sum of all episode charges is $O(1)$
by \Cref{nr:push-account}. Thus only $O(N/b)$ episodes can end this
way. This pays for the full solves requested without first meeting
the output goal.

If the goal is reached, there are two possibilities above the final
scale. When the remaining mass is below $Q/2$, the cache loop has
already made progress; we begin a smaller episode without a full
solve. Otherwise $\mu\ge Q/2$, so $\theta=Q/b\le2\mu/b$.
The output goal then supplies exactly the density bound required
for the good-case/bad-case argument, and we request a full solve.
At $\theta=q$, we instead return the final approximation.

Finally, the gap between the admission threshold $3\theta$ and the
retained level near $\theta$ controls how often the loop solves.
A call either admits a new vertex into cache or revisits an unchanged cache.
In the latter case, an already processed row needs fresh exterior
flow to reach the threshold again. Monotone admissions and globally
bounded exterior flow pay for these two kinds of call, respectively
(\Cref{nr:cache-count}). Thus the cache implements the output goal
with both its size and its repeated work under control.

\subsection{The complete algorithm}
\label{nr:details}

The implementation uses the three operations in
\Cref{nr:alg-exterior,nr:alg-full,nr:alg-cache}.
\textnormal{\textsc{ExteriorPush}} drains residual outside $H$, stopping
if an open budget is exhausted. \textnormal{\textsc{FullSolve}} solves
on $H$ and leaves at most $\tau d_v$ of work residual at each active
vertex, where $\tau=qb/(2^{12}N)$.
\textnormal{\textsc{CacheSolve}} solves on $C$ and retains the floor
specified above. The two joint solves cost
$\widetilde O(\vol(H))$ and $\widetilde O(\vol(C))$, respectively.
These costs use the SDDM convention in \Cref{sec:preliminaries}.

The state is \eqref{nr:state}. An exterior queue stores vertices outside $H$ with
$r_v>qd_v$, and an active heap stores residual densities on $H$.
Every changed row has its queue or heap key updated. Degrees are read
when a vertex first receives residual; adjacency lists are opened only
by exterior pushes and then stored. Heap operations cost $O(\log N)$.
At episode entry one may scan $H$ to initialize the high-density queue;
subsequent cache iterations update it only at touched rows.

The counter $W$ is exactly \eqref{nr:charge}. It changes only at the
exterior-push and activation lines below. Joint solves update $x$ and the residual records without adding a charge.
Disconnected solve domains are handled component by component. If the source component
becomes complete, \Cref{lem:potential-completion} takes priority over every
subsequent local operation and returns the final answer.

\begin{algorithm}[H]
\caption{\textnormal{\textsc{ExteriorPush}}: drain the exterior queue}
\label{nr:alg-exterior}
\small
\KwIn{State $(x,H,r)$; an optional open episode $(Q,W)$.}
\KwOut{Updated state and status; or a final answer if the source component is complete.}
\While{$\exists v\notin H$ with $r_v>qd_v$}{
 Choose such $v$ and read its adjacency list if not stored\;
 $t\leftarrow\min\{r_v/(5d_v),c-x_v\}$\nllabel{nr:line-push-step}\;
 $x_v\leftarrow x_v+t$; $r_v\leftarrow r_v-5d_vt$\;
 \For{$u\in N(v)$}{$r_u\leftarrow r_u+(1-\alpha) t$; update its queue or heap key\;}
 \If{an episode is open}{$W\leftarrow W+(1-\alpha) d_vt$\;}
 \If{$x_v=c$}{
 $H\leftarrow H\cup\{v\}$; $r_v\leftarrow r_v+(q/2)d_v$\;
 \If{an episode is open}{$W\leftarrow W+(3q/2)d_v$\;}
 }
 Update $v$'s exterior queue entry or active heap key\;
 \lIf{$H$ is the complete source component}{terminate the whole algorithm with the answer from \Cref{lem:potential-completion}}
 \lIf{an episode is open and $W\ge Q$}{\Return \textnormal{budget exhausted}}
}
\Return \textnormal{drained}\;
\end{algorithm}

\begin{algorithm}[H]
\caption{\textnormal{\textsc{FullSolve}}: reduce work residual on the whole active set}
\label{nr:alg-full}
\small
\KwIn{State $(x,H,r)$ and block scale $b$.}
\KwOut{State with $0\le r_H\le\tau\mathbf d_H$, where $\tau=qb/(2^{12}N)$.}
Set $\tau\leftarrow qb/(2^{12}N)$\;
Compute $h_H\ge0$ by an SDDM solve with
 $0\le r_H-M_{HH}h_H\le\tau\mathbf d_H$\;
$x_H\leftarrow x_H+h_H$; $r_H\leftarrow r_H-M_{HH}h_H$\;
\For{$v\in H$ and $u\in N(v)\setminus H$}{
 $r_u\leftarrow r_u+(1-\alpha) h_v$; update its exterior queue key\;
}
Update the active heap on $H$; retain $H,Q,W$\;
\end{algorithm}

\begin{algorithm}[H]
\caption{\textnormal{\textsc{CacheSolve}}: retain a residual floor on the cache}
\label{nr:alg-cache}
\small
\KwIn{State $(x,H,r)$, cache $C\subseteq H$ with $r_C\ge\theta\mathbf d_C$, and fixed $\theta\ge q$.}
\KwOut{State with $\theta\mathbf d_C\le r_C\le(\theta+q/4)\mathbf d_C$.}
$g_C\leftarrow r_C-\theta\mathbf d_C$\;
Compute $h_C\ge0$ by an SDDM solve with
 $0\le g_C-M_{CC}h_C\le(q/4)\mathbf d_C$\;
$x_C\leftarrow x_C+h_C$; $r_C\leftarrow r_C-M_{CC}h_C$\;
\For{$v\in C$ and $u\in N(v)\setminus C$}{
 $r_u\leftarrow r_u+(1-\alpha) h_v$; update its queue or heap key\;
}
Update the active heap on $C$; retain $H,Q,W,C$\;
\end{algorithm}

\begin{algorithm}[H]
\caption{Nonrecursive local PageRank computation}
\label{nr:alg-main}
\small
\KwIn{$G,s,0<\alpha\le1/2,\eps$ and integer $1\le b\le N=96/\eps$.}
\KwOut{An ACL $\eps$-approximation $\widehat\pi_s$.}
Query $d_s$; \lIf{$d_s=0$}{\Return $\mathbf e_s$}

Initialize $q=\eps/96$, $c=3q/8$, $x=0$, $H=\varnothing$, $r=\mathbf e_s$\nllabel{nr:line-init}\;
Run \textnormal{\textsc{ExteriorPush}} with no open episode\nllabel{nr:line-initial-drain}\;
\lIf{$H=\varnothing$}{\Return $\alpha D(x-c\one)_+$}
$\mathit{needFull}\leftarrow\textnormal{true}$\;
\While{true}{
 $H^{\rm start}\leftarrow H$; $\mu_0\leftarrow\sum_{v\in H}r_v$\nllabel{nr:line-episode-start}\;
 $Q\leftarrow\max\{\mu_0,qb\}$; $\theta\leftarrow Q/b$; $W\leftarrow0$; open the episode\;
 \lIf{$\mathit{needFull}$}{run \textnormal{\textsc{FullSolve}}$(H,r,b)$}
 $C\leftarrow\varnothing$; $\mathit{abort}\leftarrow\textnormal{false}$\nllabel{nr:line-cache-reset}\;
 \While{true}{
 Run \textnormal{\textsc{ExteriorPush}} with the open episode\nllabel{nr:line-normalize-start}\;
 \lIf{$W\ge Q$}{$\mathit{abort}\leftarrow\textnormal{true}$; break}
 $T\leftarrow\{v\in H\setminus C:r_v\ge3\theta d_v\}$; $C\leftarrow C\cup T$\;
 \lIf{$\max_{v\in H}r_v/d_v<3\theta$}{break}
 Run \textnormal{\textsc{CacheSolve}}$(C,r,\theta)$\nllabel{nr:line-cache-call}\;
 }
 Close the episode; $\mu\leftarrow\sum_{v\in H}r_v$\;
 \eIf{$\mathit{abort}$}{\nllabel{nr:line-abort}
 Run \textnormal{\textsc{ExteriorPush}} with no open episode\;
 $\mathit{needFull}\leftarrow\textnormal{true}$\;
 }{
 \lIf{$\theta=q$}{\Return $\alpha D(x-c\one)_+$}\nllabel{nr:line-final}\;
 $\mathit{needFull}\leftarrow(\mu\ge Q/2)$\nllabel{nr:line-next-full}\;
 }
}
\end{algorithm}

\Cref{nr:alg-main} puts these operations together. Initially, exterior
pushes from the source discover the first active region. Each episode
then fixes its allowance and runs the cache loop, with a full solve only
when requested by the preceding episode. A successful episode at
$\theta=q$ terminates: both the active and exterior residual targets
are met. Above this scale, success has two meanings. If the mass has
fallen below $Q/2$, the next episode uses a smaller allowance without a
full solve. Otherwise the residual is sufficiently spread out for
\Cref{nr:progress}, and the next episode begins with a full solve.
Budget exhaustion instead invokes the restart just described. These are
the decisions in lines~\ref{nr:line-abort}--\ref{nr:line-next-full}.

Only exterior pushes reveal new graph adjacency lists. The joint updates use
stored lists and record residual at boundary endpoints; they do not open
those endpoints' lists. Queues and a heap of active residual densities
are updated only at touched rows, so finding the next excessive residual
does not require rescanning $H$ after every cache solve. A new episode
may scan $H$ once. If discovery completes the source component, the
completion procedure in \Cref{lem:potential-completion} returns the final
answer immediately.

\subsection{Analysis}
\label{nr:analysis}

We first establish the discovery cost and the final residual bound,
then prove the cache and full-solve comparisons used above, and finally
count the operations. A bound on $x$ below supplies the
remaining ingredient for counting coordinate doublings.

The reserve serves two purposes: its growth pays for exterior scans,
and its bounded size leaves enough of the error allowance for the final
output. The next two lemmas establish these claims separately.

\begin{lemma}[Global discovery account]\label{nr:push-account}
Every valid finite execution prefix has $r\ge0$, $x\ge0$, and
$\vol(H)\le N$.  The total cost of exterior pushes is $O(N)$. If $S=\sum_k d_{v_k}t_k$ sum over all exterior pushes, then $S\le3/8$, and the
sum of episode charges $W$ is at most $15/8$.
\end{lemma}
\begin{proof}
All updates preserve $e_s-Mx=q\mathbf d_H+4D_{V\setminus H}x+r$. Summing over $v\in V$ and utilize $\mathbf 1^TM=\alpha \mathbf d^T$, we have
\[
1=\alpha\mathbf d^{\mathsf T}x+q\vol(H)
 +4\sum_{v\notin H}d_vx_v+\|r\|_1.
\]
Every term is nonnegative.  Hence $\vol(H)\le1/q=N$ and
$q\vol(H)+4\sum_{v\notin H}d_vx_v \le1$.

Before activation a coordinate changes only through exterior pushes, and
these pushes stop at $x_v=c$.  Thus
$S=c\vol(H)+\sum_{v\notin H}d_vx_v$.  Since $c=3q/8$, the summed identity implies 
\begin{align*}
  S&=\frac{3}{8}q\vol(H)+\sum_{v\notin H}d_vx_v\\
  &\le \frac{3}{8}\left(q\vol(H)+4\sum_{v\notin H}d_vx_v\right)\le \frac{3}{8}.
\end{align*}

To bound the total operations of exterior pushes, we notice the following: (i) for push operations that make the vertex active, the total cost is at most $\vol(H)\le N$ since each vertex is activated at most once; (ii) for push operations that do not make the vertex active, the value $t=r_v/(5d_v)>q/5$ and therefore
\[\sum_{\textnormal{non-active pushes}}d_v\le \sum_{\textnormal{non-active pushes}}\frac{5}{q}d_v t\le \frac{5}{q}S=O(N).\]
So the total cost of exterior pushes is $O(N)$.

Finally, we bound the sum of charges $W$. We denote $W_e$ be the total charge at $e$-th episode. We denote $H_e^{\textnormal{start}}$ be the active set $H$ at the start of $e$-th episode and $H_e^{\textnormal{end}}$ be $H$ at the end of $e$-th episode. We have
\begin{align*}
  \sum_e W_e&=\frac{3q}{2}\sum_e \vol(H_e^{\textnormal{end}}\setminus H_e^{\textnormal{start}})+(1-\alpha)\sum_{\textnormal{pushes when }\exists \textnormal{ episode open}}d_v t\\
  &\le \frac{3q}{2}\vol(H)+(1-\alpha)S\le \frac{3}{2}+\frac{3}{8}=\frac{15}{8}.
\end{align*}
\end{proof}

\begin{lemma}\label{nr:accuracy}
At a normal return of \Cref{nr:alg-main}, let $z=(x-c\one)_+$.
Then $0\le\mathbf e_s-Mz\le(39q/8)\mathbf d\le\eps\mathbf d$.
Consequently $\widehat\pi_s=\alpha Dz$ satisfies the ACL
$\eps$-error guarantee.
\end{lemma}
\begin{proof}
We first prove the bound on the solution $x$ before pruning. Specifically, we prove $0\le \mathbf e_s-Mx\le 4q\mathbf d$. At a normal return of \Cref{nr:alg-main}, we have $r_v\le qd_v$, $x_v\le c$ for $v\notin H$, and $r_v\le 3qd_v$, $x_v\ge c$ for $v\in H$. Therefore, by \Cref{nr:state}, $(\mathbf e_s-Mx)_v=qd_v+r_v$ for $v\in H$ and $(\mathbf e_s-Mx)_v=4d_vx_v+r_v$ for $v\notin H$. Therefore, we have $0\le (\mathbf e_s-Mx)_v<4qd_v$ for $v\in H$ and $0\le (\mathbf e_s-Mx)_v\le (4c+q)d_v\le \frac{5}{2}qd_v$ for $v\notin H$. So $0\le (\mathbf e_s-Mx)_v<4qd_v$ for both cases.

Next we prove the bound after purning. Put $a=x-z=\min\{x,c\one\}$. By
$\mathbf e_s-Mz=(\mathbf e_s-Mx)+Ma$, we just need to bound the term $Ma$. For the upper bound, $Ma\le Da\le c\mathbf d$. Thus
$\mathbf e_s-Mz\le(4q+c)\mathbf d\le(39q/8)\mathbf d$.

For the lower bound, we split the case into $v\in H$ and $v\notin H$. For $v\in H$, $a_v=c$ and all neighbors $u\in N(v)$ have $a_u\le c$, so
$(Ma)_v=d_vc-(1-\alpha)\sum_{u\in N(v)}a_u\ge\alpha cd_v\ge0$, so $(\mathbf e_s-Mz)_v\ge 0$. For $v\notin H$, $z_v=0$, and hence
$(\mathbf e_s-Mz)_v=(\mathbf e_s)_v+(1-\alpha)\sum_{u\in N(v)}z_u\ge0$. Then we apply \Cref{lem:acl-residual} to prove the ACL guarantee.
\end{proof}

We first justify the mass comparison used to bound the monotone cache.
It compares two endpoints on the final active set, so newly activated
rows are accounted for from the drained starting snapshot.

\begin{proof}[Proof of \Cref{nr:endpoint}.]
Consider the active set $H$ and the active set at the beginning of the episode $H^{\rm start}$. Let $u=x-x^0\ge0$, where $(x^0,r^0)$ is the episode's starting state, and $(x,r)$ is the current state.
We denote the sum of episode charges $W=\frac{3q}{2}\vol(H\setminus H^{\rm start})+U$, where $U$ is the aggregate outflow of exterior pushes during the episode.

We first prove $\mu\le \mu_0+W$. The residual satisfies $\mathbf e_s-Mx=(\mathbf e_s-Mx^0)-Mu$. For $v\in H^{\rm start}$, it holds that $qd_v+r_v=qd_v+r_v^0-(Mu)_v$ since the retained residual on $H^{\rm start}$ is $qd_v$. So $r_v=r_v^0-(Mu)_v$. 
For $v\in H\setminus H^{\rm start}$, the initial retained residual is $4d_vx_v^0$ and the current retained residual is $qd_v$. This yields
$r_v=(4x_v^0-q)d_v+r_v^0-(Mu)_v$.
The starting state is drained and exterior on $H\setminus H^{\rm start}$, so
$x_v^0<c=\frac{3q}{8}$ and $r_v^0\le qd_v$ for $v\in H\setminus H^{\rm start}$.  So $r_v\le (3q/2)d_v-(Mu)_v$ for $v\in H\setminus H^{\rm start}$.

Sum over $H$ gives
\[\mu=\sum_{v\in H} r_v\le \sum_{v\in H^{\rm start}} r_v^0 + \frac{3q}{2}\vol(H\setminus H^{\rm start})-\sum_{v\in H} (Mu)_v.\]

So next we just need to bound $-\sum_{v\in H} (Mu)_v\le U$. To see this, we write
\[-\sum_{v\in H} (Mu)_v=-\mathbf 1_H^T M_{HH} u_H +(1-\alpha) \sum_{w\notin H} d_H(w) u_w.\]

Denote $F=H\setminus H^{\rm start}$. The principal column sums of $M_{HH}$ are nonnegative,
so $-\one_F^{\mathsf T}M_{FF}u_F\le 0$. Since $U$ record all the outflow of exterior pushes during the episode, we have $(1-\alpha)\sum_{w\notin H} d_H(w) u_w \le (1-\alpha)\sum_{w\notin H} d(w) u_w \le U$.
 This gives
$\mu\le\mu_0+(3q/2)\vol(H)+U=\mu_0+W$.
Before exhaustion, $W<Q$ and $\mu_0\le Q$ imply $\mu<2Q$.

Next we prove $\vol(C)<2b$.
We notice by the construction of cache $C$, the residual satisfies $r_v\ge \theta d_v$ before next cache solve.
Thus $\theta\vol(C)\le \sum_{v\in C} r_v\le\mu<2Q=2\theta b$. So $\vol(C)<2b$.
\end{proof}

We next prove the boundary comparison in \Cref{nr:progress}.

\begin{proof}[Proof of \Cref{nr:progress}.]

First we recall some notations: $A$ is the active set at first \textnormal{\textsc{FullSolve}}, $B$ is newly added nodes before next \textnormal{\textsc{FullSolve}}, $x^0,r^0$ are the solution and residual before the first \textnormal{\textsc{FullSolve}}, $h$ is the solution of the first \textnormal{\textsc{FullSolve}}, $f$ is the work residual of the first \textnormal{\textsc{FullSolve}}. We denote $D=\{v\in A:h_v>x_v^0\}$ and $w=(h-x_A^0)_+$.

We compare the state immediately after the first
\textnormal{\textsc{FullSolve}} with the next full-solve state on the
final active set $A\cup B$. Using \eqref{nr:state}, we bound the
contributions to $\mu_{\rm next}$ in four parts:
\begin{enumerate}[label=(\roman*),leftmargin=*,itemsep=3pt]
\item The work residual left on the old vertices $A$ immediately after
the first \textnormal{\textsc{FullSolve}}. This is $f$, so its
contribution is $\|f\|_1$.

\item The initial work residual on the new vertices $B$, together with
the adjustment from their exterior reserves to active reserves.
The initial state is drained, so $r_v^0\le qd_v$ for each $v\in B$.
The reserve adjustment in the endpoint comparison contributes at most 
$\frac q2d_v$.
Thus this part is at most $(3q/2)\vol(B)$.

\item The residual sent from $A$ to $B$ by the first
\textnormal{\textsc{FullSolve}}. This is its boundary flow
$(1-\alpha)\one_B^{\mathsf T}A_{BA}h$, which we bound later.

\item The residual subsequently sent into $A\cup B$ from outside
$A\cup B$. Those exterior vertices remain outside the active set
throughout the interval, so this mass is at most $U$.
\end{enumerate}

Consequently,
\[
 \mu_{\rm next}\le\|f\|_1+\frac{3q}{2}\vol(B)
       +(1-\alpha)\one_B^{\mathsf T}A_{BA}h+U.
\]

It therefore remains to bound the contribution $(1-\alpha)\one_B^{\mathsf T}A_{BA}h$.
Since $h\le x_A^0+w$ and $w$ is supported on $D$, we split it into
an ordinary part and an excess part:
\[
 (1-\alpha)\one_B^{\mathsf T}A_{BA}h
 \le (1-\alpha)\one_B^{\mathsf T}A_{BA}x_A^0
       +(1-\alpha)\one_B^{\mathsf T}A_{BD}w_D.
\]

For the ordinary part, the initial exterior-state condition ($x_v^0<c,r_v^0\le qd_v$) at each
$v\in B$ gives
\[
 (1-\alpha) A_{vA}x_A^0\le5d_vx_v^0+r_v^0\le(5c+q)d_v.
\]
Summing over $B$, the ordinary boundary flow is at most
$(5c+q)\vol(B)$.

For the excess part, subtract the initial active-state equation 
\[M_{AA}x^0_A=(\mathbf e_s)_A+(1-\alpha)A_{A,V\setminus A}x^0_{V\setminus A}-q\mathbf d_A-r_A^0\]
 from
$M_{AA}h=r_A^0-f$. Since $f$ and $x^0$ are
nonnegative, this gives the inequality
\[
 M_{AA}(h-x_A^0)\le q\mathbf d_A+2r_A^0.
\]

On $v\in D$, we have that
\[[M_{AA}(h-x_A^0)]_v=(M_{DD}w_D)_v-(1-\alpha)\sum_{u\in A\setminus D} A_{vu}(h_u-x_u^0)\ge (M_{DD}w_D)_v.\]
Hence
$M_{DD}w_D\le q\mathbf d_D+2r_D^0$, and
\[
 (1-\alpha)\one_B^{\mathsf T}A_{BD}w_D
 \le\one_D^{\mathsf T}M_{DD}w_D
 \le q\vol(D)+2\sum_{v\in D} r_v^0.
\]
The first inequality holds because the middle expression includes
all flow leaving $D$, as well as the nonnegative teleportation
contribution.

Combining the four contributions with these two boundary-flow bounds,
and using $c=3q/8$, we obtain
\begin{align*}
 \mu_{\rm next}
 &\le\|f\|_1+U+
       \left(\frac{3q}{2}+5c+q\right)\vol(B)
       +q\vol(D)+2\sum_{v\in D} r_v^0\\
 &\le J+q\vol(D)+2\sum_{v\in D} r_v^0+\|f\|_1,
\end{align*}
where $J=12q\vol(B)+U$. This proves the Lemma.
\end{proof}

The inequality becomes a contraction once the starting density is
controlled. The following quantitative form is used in the episode count.

\begin{corollary}\label{nr:halving}
Under the notation of \Cref{nr:progress}, suppose $\mu>qb/2$,
$r_v^0\le6\mu d_v/b$ on $A$, and $\|f\|_1\le qb/2^{12}$.
If $J<\mu/4$ and $\vol(D)\le b/192$, then
$\mu_{\rm next}<\mu/2$.
\end{corollary}
\begin{proof}
Since $\mu>qb/2$, we have $q<2\mu/b$ and
$\|f\|_1\le qb/2^{12}<\mu/2048$.
Thus $q\vol(D)+2r^0(D)<14\mu\vol(D)/b\le14\mu/192$.
Together with $J<\mu/4$ and the error bound, this is less than $\mu/2$.
\end{proof}

To count large-growth events, we need a bound on how often an active
coordinate can double. The following fact uses the state identity \eqref{nr:state}
and the boundary of the discovered region.

\begin{lemma}
\label{lem:potential-completion}
Whenever $H$ is nonempty, it is connected and contains $s$.
As long as $H$ is a proper subset of the source component, every
potential satisfies $0\le x_v\le3N$.
\end{lemma}
\begin{proof}
A nonsource vertex can activate only if it has an active neighbor: at activation $x_v=c$ and its residual $(\mathbf e_s-Mx)_v$ is at least $3qd_v/2$, whereas neighbors all below $c$ could supply at most $(1-\alpha)cd_v$. Thus the first active vertex is $s$, and every subsequent one attaches to $H$. If $H$ is a proper subset of the source component, the grounded Laplacian satisfies $(L^{-1}_{HH})_{uu}\le |H|\le N$ for all $u\in H$.
By the property of SDDM matrix, we directly have
$0\le(M_{HH}^{-1})_{uv}\le (M_{HH}^{-1})_{uu} \le \frac{1}{1-\alpha}(L^{-1}_{HH})_{uu} \le 2N$ for all $u,v\in H$.
By the state identity and the exterior bound $x_v<c$,
$M_{HH}(x_H-c\one_H)\le\mathbf e_s|_H$.
Multiplying by $M_{HH}^{-1}$ yields $x_v\le c+2N\le3N$ on $H$.
Outside $H$, $x_v<c$ already proves the claim.

If no edge leaves $H$, connectivity shows that it is the whole
source component. Every active vertex already has its adjacency
list stored. Maintaining the number of edges leaving $H$ costs
$O(d_v)$ when $v$ activates and detects this event without a
global graph scan. Solve $M_{HH}p_H=\mathbf e_s|_H$, set
$p=0$ outside $H$, and return $\alpha Dp$. There are no boundary
edges, so $Mp=\mathbf e_s$ and this answer is exact. The solve
costs $\widetilde O(\vol(H))\le\widetilde O(N)$ under the
convention of \Cref{sec:preliminaries}. This completion test
takes priority over further local operations.
\end{proof}

The full-solve comparison is useful only after successful normalization.
We therefore count budget restarts separately, then count the large-input
and large-growth exceptions to halving. All statements below apply to
finite valid execution prefixes; termination is a consequence of the
counts, not an assumption behind them.

\begin{lemma}\label{nr:full-exceptions}
Consider consecutive calls of \textnormal{\textsc{FullSolve}} in
\Cref{nr:alg-main}. Let $\mu$ and $\mu_{\rm next}$ be the active
work-residual masses immediately before the two calls, and let $J$
and $D$ be as in \Cref{nr:progress}. If $\mu_{\rm next}\ge\mu/2$,
then at least one of the following events occurs:
\begin{enumerate}[label=(\arabic*),leftmargin=*,itemsep=3pt]
\item The first \textnormal{\textsc{FullSolve}} is the initial call,
or was requested after an episode exhausted its budget, $W\ge Q$.
\item The first \textnormal{\textsc{FullSolve}} follows a successful
episode, and $J\ge\mu/4$.
\item The first \textnormal{\textsc{FullSolve}} follows a successful
episode, and $\vol(D)>b/192$.
\end{enumerate}
Over any finite valid execution prefix, the total number of occurrences
of these events is $O((N/b)\log N)=\widetilde O(N/b)$.
\end{lemma}
\begin{proof}
We first prove if neither the three event happens, then $\mu_{\rm next}<\mu/2$. If (1) does not occur, the first full solve was requested by a successful
episode whose budget $Q$ satisfies $Q>qb$ and $\mu\ge Q/2$. Its output is drained so $\mu>qb/2$, and the density guarantee gives $r_v^0<\frac{3Q}{b}d_v\le\frac{6\mu}{b}d_v$ for $v\in H$.
The full solve leaves
$\|f\|_1\le\tau\vol(H)\le qb/2^{12}$.
If neither (2) nor (3) occurs, \Cref{nr:halving} therefore gives
$\mu_{\rm next}<\mu/2$. 

We next count the three events separately. For (1), every exhausted episode
has spent $W\ge Q\ge qb$, whereas \Cref{nr:push-account} gives
$\sum W\le15/8$ over all episodes. So
\[
 \#\{\textnormal{Exhausted episodes}\} \le \frac{15}{8qb}=O(N/b).
\]
Each requests at most one full solve; adding the initial call gives $O(N/b)$ occurrences of (1).

For (2), index the disjoint intervals between consecutive full solves
by $i$, and write $J_i=12q\vol(B_i)+U_i$. Each vertex is newly activated
in at most one interval, and the total exterior-push outflow is at most
$\sum_i U_i\le 3/8$ by \Cref{nr:push-account}. Hence
\[
 \sum_i J_i
 \le12q\sum_i\vol(B_i)+\sum_iU_i
 \le12qN+\frac38=\frac{99}{8}.
\]
Whenever (2) occurs, the successful starting state has $\mu>qb/2$,
so $J_i\ge\mu/4>qb/8$. Consequently, (2) occurs only $O(N/b)$ times.

For (3), each vertex enters the active set at potential
$c=3/(8N)$. On the other hand,
\Cref{lem:potential-completion} gives $x_v\le3N$ before termination.
Whenever $v\in D$, the full-solve increment satisfies $h_v>x_v$,
so the update more than doubles $x_v$. Thus, if $v$ belongs to $D$ at $k$ full solves, $2^kc\le3N$ and therefore $k=O(\log N)$.
Writing $D_i$ for the growth set at full solve $i$, we obtain
\[
 \sum_i\vol(D_i)
 =\sum_{v\text{ activated}}d_v\,\#\{i:v\in D_i\}
 =O(N\log N),
\]
because the total activated volume is at most $N$.
Each occurrence of (3) contributes more than $b/192$ to this sum,
so there are $O((N/b)\log N)$ such occurrences.
Adding the three bounds proves the claim.
\end{proof}

\begin{lemma}[Full solves and episodes]\label{nr:counts}
For fixed $b$, \Cref{nr:alg-main} performs
$\widetilde O(N/b)$ full solves and $\widetilde O(N/b)$ episodes.
\end{lemma}
\begin{proof}
Call a pair of consecutive \textnormal{\textsc{FullSolve}}s exceptional if any of
(1)--(3) in \Cref{nr:full-exceptions} occurs.
For every other pair, that lemma gives
$\mu_{\rm next}<\mu/2$. Since (1) does not occur, its first \textnormal{\textsc{FullSolve}}
follows a successful episode and has $\mu>qb/2$.
Also, \Cref{nr:push-account} gives $\mu\le\|r\|_1\le1$.
Thus the starting masses in a consecutive run of nonexceptional pairs
satisfy $qb/2<\mu\le1$ and decrease by a factor of at least two from one pair to the next.
So we can only have $O(\log N)$ consecutive nonexceptional pairs.
Partitioning the \textnormal{\textsc{FullSolve}} sequence at the
$\widetilde O(N/b)$ exceptional pairs gives the upper bound of $ \#\textnormal{\textsc{FullSolve}}=\widetilde O(N/b)$.

It remains to count episodes. We consider the following three cases: (i) the episode is ended with a \textnormal{\textsc{FullSolve}} request; (ii) the episode is ended with an exhausted budget; (iii) the episode success without requesting a \textnormal{\textsc{FullSolve}}. For (i), the number of such episodes is at most $\{\#\textnormal{\textsc{FullSolve}}+1\}=\widetilde O(N/b)$. For (ii), the number of such episodes is at most $O(N/b)$ by \Cref{nr:full-exceptions}. For (iii), output of such episode satisfies $\mu<Q/2$, and the budget of the next episode is set as $Q'=\max\{\mu,qb\}\le \max\{Q/2,qb\}$.
Since $qb\le Q\le1$, a consecutive run of these episodes has length
$O(\log(N))$: the budget halves until it reaches $qb$, at which
a successful episode terminates the algorithm. As a result,
\[\#\{\textnormal{Episodes}\} \le O(\log N) \times (\#\{\textnormal{\textsc{FullSolve}}\} + \#\{\textnormal{Exhausted episodes}\} + 1) = \widetilde O(N/b).\]
\end{proof}

\begin{lemma}[\textnormal{\textsc{CacheSolve}} count]\label{nr:cache-count}
If $E$ is the number of episodes, the total number of \textnormal{\textsc{CacheSolve}}s is
$O(bE+N)$.
\end{lemma}
\begin{proof}

We split the calls to \textnormal{\textsc{CacheSolve}} into two cases: (i) the call admits new vertices comparing with the privious \textnormal{\textsc{CacheSolve}}; (ii) the call does not admit new vertices. For (i), the number of such calls is at most $O(bE)$ beacuse $\vol(C)\le 2b$ by \Cref{nr:endpoint}.

Now we bound the number of case (ii) \textnormal{\textsc{CacheSolve}}. It must have a previous
cache solve in the same episode, because the cache is initially empty.  Immediately
after that solve, all cached residuals $r_v\le (\theta+q/4)d_v$.  For another solve to be requested, a cached vertex $v$ satisfies $r_v\ge 3\theta d_v$, requiring an increase of at least
$(2\theta-q/4)d_v\ge q$. Therefore every case (ii) \textnormal{\textsc{CacheSolve}} needs at least $q$ new residual mass.
We notice that the only way to get the $q$ new residual mass is by \textnormal{\textsc{ExteriorPush}}. However by \Cref{nr:push-account}, the total outflow of \textnormal{\textsc{ExteriorPush}} is at most $3/8$. So the number of case (ii) \textnormal{\textsc{CacheSolve}} is at most $3/(8q)=O(N)$.
Adding the two accounts proves the claim.
\end{proof}

\begin{proof}[Proof of \Cref{nr:runtime}.]
By \Cref{nr:counts}, the numbers of \textnormal{\textsc{FullSolve}}s and episodes are both
$\widetilde O(N/b)$. Each full solve and each scan at each episode costs $\widetilde O(N)$, since \Cref{nr:push-account} gives
$\vol(H)\le N$. Their total cost is therefore $\widetilde O(N^2/b)$.

If $E$ is the number of episodes, \Cref{nr:cache-count} bounds the
number of cache solves by $O(bE+N)=\widetilde O(N)$.
Every such solve has volume $\vol(C)\le 2b$ by \Cref{nr:endpoint}, so their total cost is $\widetilde O(Nb)$.

The total operations of exterior pushes are bounded by $O(N)$ by
\Cref{nr:push-account}. Summing the costs gives $\widetilde O(N^2/b+Nb)$ expected time.
The finite-prefix counts also imply termination: each drain is
finite by \Cref{nr:push-account}, each cache loop is finite by
\Cref{nr:cache-count} with one episode, and \Cref{nr:counts}
bounds the number of episodes. At an ordinary return,
\Cref{nr:accuracy} gives the ACL guarantee. If the source
component is completed first, \Cref{lem:potential-completion}
returns the exact answer with an additional $\widetilde O(N)$
work, which is absorbed in the bound.
\end{proof}

The remaining cost now has a precise source. The cache loop performs
many direct solves on sets of volume $O(b)$, and these account for the
term $Nb$. Each such loop is trying to accomplish the same residual
normalization that we have just learned to perform on the larger active
set. The next section uses this fact to replace a normalization attempt
by a smaller recursive computation.

\section{Extending the method to recursion}
\label{sec:recursive-reduction}

The nonrecursive algorithm has separated two tasks: discovering the
active region and deciding when to solve on all of it. Its remaining
$Nb$ term comes from repeatedly solving on a cache of volume $O(b)$.
We can reduce this cost by recursively applying the same method inside that smaller
region, rather than calling SDDM solver on the whole cache. The normalization task passed to the next level is the \emph{output goal}
from \Cref{nr:full-idea}: reduce the residual densities throughout the
parent's active set, using joint solves only on a smaller discovered
set. One parent episode makes one child call for this entire task.
The child in turn uses discovery, occasional full solves on its own
active set, and still smaller normalization tasks.

The parent's decisions remain those of \Cref{nr:alg-main}. Successful
normalization either finishes the problem, reduces the mass enough
to skip a full solve, or supplies the density bound needed to justify
one. Consequently the good-case/bad-case argument can also be reused.
The key is to ensure that normalizing the child's work residual also
normalizes the residual seen by its parent. We state the two conditions
that make this implication valid and derive the main result before
giving the implementation in \Cref{rec:algorithm}; the detailed proofs
are in \Cref{sec:technical-proofs}.

\subsection{Recursive guarantees and the main result}
\label{rec:reduction}
\label{rec:interface}
\label{rec:execution}

Consider a parent episode with budget $Q$, cache scale $b$, and target
$\theta=Q/b$. A level-$j$ child receives the parent's work residual
$r^{(j-1)}$ on $H_{j-1}$, uses threshold $q_j=\theta$, and discovers
its own active set $H_j$. All levels update the same vector $x$.
To measure the child's contribution, record $\xi_v^{(j)}=x_v$ when
$v$ first becomes available and use the \emph{relative potential}
$y_v^{(j)}=x_v-\xi_v^{(j)}$. This reference stays fixed across the
episodes of the call. The child starts with $H_j=\varnothing$,
$y^{(j)}=0$, and $r^{(j)}=r^{(j-1)}$ on its domain, and uses the
same reserve and activation rule as in \Cref{nr:alg-exterior}.
The parent--child state is
\begin{align}
 r_v^{(j-1)}=
 \begin{cases}
 q_jd_v+r_v^{(j)}, & v\in H_j,\\
 4d_vy_v^{(j)}+r_v^{(j)}, & v\in H_{j-1}\setminus H_j.
 \end{cases}
 \label{rec:parent-child}
\end{align}

The output goal concerns the parent's residual throughout
$H_{j-1}$. However, bounding the child's work residual alone is insufficient:
outside $H_j$, the parent also sees the reserve $4d_vy_v^{(j)}$.
We therefore need both a work-residual bound and a bound on the
relative potential at inactive vertices. Inside $H_j$, the only
additional contribution is the retained floor $q_jd_v$.

Fix a maximum depth $\ell\ge0$, with levels $0,\ldots,\ell$, and set
$\kappa_j=\ell-j+3$. The leaf uses the nonrecursive factor $3$;
each return adds one retained floor, giving
$\kappa_j=1+\kappa_{j+1}$. At the root, set $H_{-1}=V$,
$q_0=\eps/(32\kappa_0)$, $N_0=1/q_0$, and $y^{(0)}=x$.
The root state is \eqref{nr:state} with threshold $q_0$.
A \emph{normal return} means that the call has met its target.
An exhausted ancestor budget instead causes an interrupted return,
which retains the work already performed but makes no accuracy claim.

\begin{condition}[Residual at return]\label{condition:normalize-residual}
On a normal return of \textnormal{\textsc{Normalize}} at level $j$,
\begin{align*}
 (1)\qquad &0\le r_v^{(j)}\le\kappa_jq_jd_v,
       &&v\in H_j,\\
 (2)\qquad &0\le r_v^{(t)}\le q_td_v,
       &&v\in H_{t-1}\setminus H_t,\quad 0\le t\le j.
\end{align*}
\end{condition}

\begin{condition}[Relative potential at return]\label{condition:normalize-solution}
On a normal return at level $j$, $0\le y_v^{(j)}\le3q_j/8$ for
$v\in H_{j-1}\setminus H_j$, and $y_v^{(j)}\ge3q_j/8$ for $v\in H_j$.
\end{condition}

Together, the two conditions give the parent exactly the output
goal it needs. For a child at level $j\ge1$ with $q_j=\theta$,
\eqref{rec:parent-child} bounds the parent density by
$(1+\kappa_j)\theta$ on $H_j$ and by
$(4\cdot3/8+1)\theta=5\theta/2$ on $H_{j-1}\setminus H_j$.
Both are at most $\kappa_{j-1}\theta$. This is why one additional
retained floor per level suffices.

Part (2) of \Cref{condition:normalize-residual} also covers earlier
levels: a child update can trigger exterior work at an ancestor,
so the parent must resume with those queues drained as well.
\Cref{lem:query-push} supplies these exterior bounds and the
relative-potential bounds; \Cref{lem:recursive-leaf} supplies the
active residual bound at the leaf. The preceding parent--child
calculation then gives the induction in \Cref{rec:return-proof}.

\begin{lemma}[Normalization composes across levels]\label{lem:recursive-root}
Every normal return of \textnormal{\textsc{Normalize}} at level
$0\le j\le\ell$ satisfies
\Cref{condition:normalize-residual,condition:normalize-solution}.
\end{lemma}

The conditions also connect normalization to the running-time
analysis. The inherited density bound controls the residual on
the doubling set in \Cref{lem:residual_next_bound}, allowing
\Cref{lem:technical-full-halving} to reuse the nonrecursive
good-case/bad-case argument. The lower bound on active relative
potentials in \Cref{condition:normalize-solution} supplies the
starting scale for the doubling count in
\Cref{lem:technical-full-exceptions}. Together with the budget and
new-input bounds, these estimates yield the episode and full-solve
counts in \Cref{lem:recursive-layer}.

The child is also small enough for recursion to save work. While
the parent budget is unexhausted, its residual mass is below $2Q$
by \Cref{lem:mu_t-W_t}. Since each child-active vertex retains
$\theta d_v$ in \eqref{rec:parent-child}, we have
$\theta\vol(H_j)<2Q$, hence $\vol(H_j)<2b$; this is
\Cref{lem:query-push-volume}. The interruption rule prevents any
further child solve after the parent allowance is exhausted.

We also need a work guarantee. A nonroot call of scale $N$ is opened
inside one parent episode with $Q_{j-1}=q_jN$. Its own episodes use
a fixed block scale $1\le b\le N$. Let $R_j(N)$ denote the work
charged to this call and its descendants. A push is charged to the
call at which it is produced, including its updates to other open
records. Initializing a child's records on the parent region is
charged to the parent episode. This convention counts every operation
once, even when an ancestor push interrupts a child.
Write $P=C(2+\log(2/\eps)+\ell)^c$ for sufficiently large absolute
constants $C,c$.

\begin{lemma}[Work at one recursive level]\label{lem:recursive-layer}
A level-$j$ call of \textnormal{\textsc{Normalize}} has at most
$PN/b$ episodes and full solves, and at most one child call per
episode. For $j<\ell$, its expected work satisfies
\begin{align}
 R_j(N)\le P\left(\frac{N^2}{b}+\frac Nb R_{j+1}(b)\right)+O(PN).
 \label{rec:recurrence}
\end{align}
At the leaf, $R_\ell(N)\le P(N^2/b+Nb)$.
\end{lemma}

The first term pays for full solves, episode-entry scans, and child
initialization. The second pays for one smaller normalization per
episode. The leaf retains the nonrecursive cache cost.
\Cref{sec:technical-layer,sec:technical-leaf} prove these bounds
using the same halving and doubling accounts as in \Cref{nr:analysis}.

If a child costs $b^{1+1/k}$ up to the factor $P$, the two terms
in \eqref{rec:recurrence} are $N^2/b$ and $Nb^{1/k}$. Balancing
them improves the exponent to $1+1/(k+1)$. Starting from $3/2$ at
the leaf gives $4/3,5/4$, and so on. A slowly growing depth brings
the exponent to $1+o(1)$ while keeping the overhead subpolynomial.

\begin{lemma}[Recursive running time]\label{lem:recursive-complexity}
Set $K=1/\eps$ and
$\ell=\lfloor\sqrt{\log(2K)/\log\log(4K)}\rfloor$.
With $b=\lfloor N^{(\ell-j+1)/(\ell-j+2)}\rfloor$ at level $j$,
the root call has expected running time
\[
 K\exp\!\left(O\!\left(\sqrt{\log(2K)\log\log(4K)}\right)\right)
 =K^{1+o(1)}.
\]
\end{lemma}

\begin{proof}
Absorb the $O(PN)$ term into $PN^2/b$ by increasing $P$.
Put $r=\ell-j$ and choose $b=\lfloor N^{(r+1)/(r+2)}\rfloor$.
We prove by induction on $r$ that
$R_j(N)\le(4P)^{r+1}N^{1+1/(r+2)}$.
For $r=0$, \Cref{lem:recursive-layer} gives
$R_\ell(N)\le3PN^{3/2}\le4PN^{3/2}$.
Assume the claim holds for $r=k$, where $k\ge0$.
For $r=k+1$, the induction hypothesis gives
\begin{align*}
 R_j(N)
 &\le P\left(\frac{N^2}{b}+(4P)^{k+1}N b^{1/(k+2)}\right)\\
 &\le P\bigl(2+(4P)^{k+1}\bigr)N^{1+1/(k+3)}
 \le(4P)^{k+2}N^{1+1/(k+3)}.
\end{align*}
Finally, $N_0=32\kappa_0K$ and the chosen depth give
\begin{align*}
 R_0(N_0)
 &\le N_0\exp\!\left(\frac{\log N_0}{\ell+2}
                   +O((\ell+1)\log P)\right)\\
 &\le K\exp\!\left(O\!\left(\sqrt{\log(2K)\log\log(4K)}\right)\right).
\end{align*}
\end{proof}

The return guarantees now give the approximation guarantee directly.

\begin{proof}[Proof of \Cref{thm:main}]
Initialize the root with $x=0$, $H_0=\varnothing$, and
$r^{(0)}=\mathbf e_s$, and run \textsc{Normalize} at scale $N_0$
and threshold $q_0$. By \Cref{lem:recursive-layer}, it terminates.
If the source component is completed, \Cref{lem:potential-completion}
returns the exact answer with an additional $\widetilde O(N_0)$ work.
Otherwise, the root has no ancestor that can interrupt it, so its
return is normal and \Cref{lem:recursive-root} applies.

Set $c_0=3q_0/8$, $z=(x-c_0\one)_+$, and
$\widehat\pi_s=\alpha Dz$. The root state and the two return
conditions give
\[
 \mathbf e_s-Mz
 \le\bigl(\max\{q_0,4c_0\}+c_0+\kappa_0q_0\bigr)\mathbf d
 =\left(\kappa_0+\frac{15}{8}\right)q_0\mathbf d
 \le\eps\mathbf d.
\]
For $v\in H_0$, $(x-z)_v=c_0$ and
$0\le x-z\le c_0\one$ imply $(M(x-z))_v\ge0$, so
$(\mathbf e_s-Mz)_v\ge q_0d_v$. For $v\notin H_0$, $z_v=0$
implies $(\mathbf e_s-Mz)_v\ge0$. Thus
$0\le\mathbf e_s-Mz\le\eps\mathbf d$, and
\Cref{lem:acl-residual} proves the ACL guarantee.
\Cref{lem:recursive-complexity} supplies the claimed running time.
\end{proof}

It remains to implement \textsc{Normalize} and establish the two
guarantees used above.

\subsection{The complete recursive algorithm}
\label{rec:algorithm}

\Cref{alg:recursive-normalization} reuses the episode loop of
\Cref{nr:alg-main}, replacing the cache loop by one child call at
each internal level. An episode fixes
$Q_j=\max\{\mu_j^{\rm start},q_jb\}$ and $\theta=Q_j/b$, performs
a pending full solve, and asks the child to meet the output goal at
threshold $q_{j+1}=\theta$. The leaf uses the original cache loop.
After success, the same decisions apply: return if $\theta=q_j$,
skip the next full solve if the mass is below $Q_j/2$, and otherwise
request that solve with the density bound just obtained
(lines~\ref{rec:line-child-init}--\ref{rec:line-next-full}).

The additional issue is discovery across several open levels.
Ancestor joint solves are suspended during a child call, but ancestor
exterior pushes still process boundary flow. The \emph{budget counter}
$W_j$ records two kinds of new input to the level-$j$ active region:
\begin{itemize}[leftmargin=*,itemsep=3pt,topsep=3pt]
\item When $v$ enters $H_j$, add $(j+3)q_jd_v/2$ to $W_j$.
This allowance covers a vertex that began outside the parent domain
and passed through earlier inactive layers before reaching level $j$.
\item A push of increment $h$ at level $t$ adds its outgoing flow
$(1-\alpha)d_vh$ to every open $W_j$ with $j\ge t$.
Its source is outside those active sets, so this flow may increase
their residual mass.
\end{itemize}
Deeper updates are supported inside $H_j$ and cannot increase its
total mass after internal transfers cancel. They incur no direct
budget charge. \Cref{lem:mu_t-W_t} proves
$\mu_j\le\mu_j^{\rm start}+W_j$, the bound that keeps every child
small while the parent budget remains unexhausted.

The update crossing a budget is completed and charged; no further
solve is then allowed at that level or below. The owner closes its
exhausted episode, discards descendant records, drains with the
ancestor counters still open, and restarts with a full solve. If an
ancestor owns the exhausted budget, the current call returns control
to it. All increments to $x$ are retained. This is the restart rule
of \Cref{nr:cache-idea}, applied at the level whose allowance expired.

The shared state consists of the root identity \eqref{nr:state},
the identities \eqref{rec:parent-child}, and the open episode records.
Write $\mathcal S_j=(x,\{H_t,\xi^{(t)},y^{(t)},r^{(t)}\}_{t\le j})$
for these records. The active sets are nested and all actual work
residuals are nonnegative. A newly available row receives
$\xi_v^{(j)}=x_v$, $y_v^{(j)}=0$, and $r_v^{(j)}=r_v^{(j-1)}$.
The root relation $y^{(0)}=x$ is an alias, so it is not incremented
twice. References remain fixed until their call ends.

\begin{algorithm}[!htbp]
\caption{Recursive local PageRank computation}
\label{alg:recursive-normalization}
\small
\SetKwProg{RecFn}{Procedure}{:}{}
\KwIn{$G,s,0<\alpha\le1/2,\eps$ and depth $\ell$.}
\KwOut{An ACL $\eps$-approximation $\widehat\pi_s$.}
\RecFn{\textnormal{\textsc{LocalPageRank}}$(G,s,\alpha,\eps,\ell)$}{
  Query $d_s$; \lIf{$d_s=0$}{\Return $\mathbf e_s$}

  Set $\kappa_j=\ell-j+3$, $q_0=\eps/(32\kappa_0)$, and $N_0=1/q_0$\;
  Initialize $x=0$, $y^{(0)}=x$, $\xi^{(0)}=0$, $H_0=\varnothing$, $r^{(0)}=\mathbf e_s$, with no open episode\;
  $\mathcal S_0\leftarrow\textnormal{\textsc{Normalize}}(0,N_0,q_0,\mathcal S_0)$\;
  \Return $\widehat\pi_s=\alpha D(x-3q_0\one/8)_+$\;
}
\BlankLine
\RecFn{\textnormal{\textsc{Normalize}}$(j,N,q_j,\mathcal S_j)$}{
  $\mathcal S_j\leftarrow\textnormal{\textsc{ExteriorPush}}(j,\mathcal S_j)$\;
  \lIf{an open ancestor budget is exhausted}{\Return interrupted $\mathcal S_j$}
  \lIf{$H_j=\varnothing$}{\Return normal $\mathcal S_j$}
  Set $b\leftarrow\lfloor N^{(\ell-j+1)/(\ell-j+2)}\rfloor$ and $\mathit{needFull}\leftarrow\textnormal{true}$\;
  \While{\textnormal{true}}{
    $\mu_0\leftarrow\sum_{v\in H_j}r^{(j)}_v$; store the episode-entry set $H_j^{\rm start}$\;
    Set $Q_j\leftarrow\max\{\mu_0,q_jb\}$, $\theta\leftarrow Q_j/b$, $W_j\leftarrow0$; open this episode\nllabel{rec:line-episode-open}\;
    \lIf{$\mathit{needFull}$}{$\mathcal S_j\leftarrow\textnormal{\textsc{FullSolve}}(j,N,b,q_j,\mathcal S_j)$}
    \eIf{$j=\ell$}{
      $\mathcal S_j\leftarrow\textnormal{\textsc{LeafNormalize}}(j,b,\theta,\mathcal S_j)$\;
    }{
      Initialize $H_{j+1}=\varnothing$, $q_{j+1}=\theta$, $\xi^{(j+1)}=x$, $y^{(j+1)}=0$, $r^{(j+1)}=r^{(j)}$ on $H_j$, with no open child episode\nllabel{rec:line-child-init}\;
      $\mathcal S_{j+1}\leftarrow\textnormal{\textsc{Normalize}}(j+1,b,\theta,\mathcal S_{j+1})$\nllabel{rec:line-child-call}\;
    }
    \If{an open ancestor budget is exhausted}{
      Close this episode and discard descendant records; \Return interrupted $\mathcal S_j$\;
    }
    \If{$W_j<Q_j$}{
      $\mathcal S_j\leftarrow\textnormal{\textsc{ExteriorPush}}(j,\mathcal S_j)$ \tcp*{normal child/leaf return has drained these queues}
    }
    \If{an open ancestor budget is exhausted}{
      Close this episode and discard descendant records; \Return interrupted $\mathcal S_j$\;
    }
    \eIf{$W_j\ge Q_j$}{
      Close this episode and discard descendant records\;
      $\mathcal S_j\leftarrow\textnormal{\textsc{ExteriorPush}}(j,\mathcal S_j)$ \tcp*{local budget closed; ancestors remain open}
      \lIf{an open ancestor budget is exhausted}{\Return interrupted $\mathcal S_j$}
      $\mathit{needFull}\leftarrow\textnormal{true}$; \textbf{continue} the while loop\;
    }{
      $\mu\leftarrow\sum_{v\in H_j}r^{(j)}_v$; close this episode and discard descendant records\;
      \lIf{$\theta=q_j$}{\Return normal $\mathcal S_j$}
      $\mathit{needFull}\leftarrow(\mu\ge Q_j/2)$\nllabel{rec:line-next-full}\;
    }
  }
}
\end{algorithm}

\begin{algorithm}[!htbp]
\caption{\textnormal{\textsc{ExteriorPush}}: exterior discovery on the open stack}
\label{alg:query-push}
\small
\KwIn{Level $j$, shared records $\mathcal S_j$, and open episode counters.}
\KwOut{Updated records; interrupted if an open budget is exhausted.}
\lIf{some open episode $i\le j$ has $W_i\ge Q_i$}{\Return interrupted $\mathcal S_j$}
\While{$\exists (t,v)$, $0\le t\le j$, $v\in H_{t-1}\setminus H_t$, with $r_v^{(t)}>q_td_v$}{
  Choose such $(t,v)$ and read its adjacency list if not stored\;
  Set $h\leftarrow\min\{r_v^{(t)}/(5d_v),3q_t/8-y_v^{(t)}\}$\;
  $x_v\leftarrow x_v+h$; $y_v^{(i)}\leftarrow y_v^{(i)}+h$ for $1\le i\le t$\;
  $r_v^{(i)}\leftarrow r_v^{(i)}-d_vh$ for $0\le i<t$; $r_v^{(t)}\leftarrow r_v^{(t)}-5d_vh$\label{line:qp-own}\;
  \For{$u\in N(v)$ and $0\le i\le j$ with $u\in H_{i-1}$}{
    $r_u^{(i)}\leftarrow r_u^{(i)}+(1-\alpha)h$\label{line:qp-neighbor}\;
  }
  \For{each open episode $k\ge t$}{
    $W_k\leftarrow W_k+(1-\alpha)d_vh$\label{line:qp-flow}\;
  }
  \If{$y_v^{(t)}=3q_t/8$}{
    $H_t\leftarrow H_t\cup\{v\}$; $r_v^{(t)}\leftarrow r_v^{(t)}+q_td_v/2$\label{line:qp-activate}\;
    \lIf{episode $t$ is open}{$W_t\leftarrow W_t+(t+3)q_td_v/2$\label{line:qp-birth}}
    \lIf{$t<j$}{register $\xi_v^{(t+1)}=x_v$, $y_v^{(t+1)}=0$, $r_v^{(t+1)}=r_v^{(t)}$}
    \lIf{$t=0$ and $H_0$ is the complete source component}{finish globally using \Cref{lem:potential-completion}}
  }
  Update queue and density keys at all touched rows\;
  \lIf{some open episode $i\le j$ has $W_i\ge Q_i$}{\Return interrupted $\mathcal S_j$}
}
\Return $\mathcal S_j$\;
\end{algorithm}

\begin{algorithm}[!htbp]
\caption{\textnormal{\textsc{FullSolve}}: joint update on the current active set}
\label{alg:full-solve-interface}
\small
\KwIn{$j,N,b,q_j$ and shared records $\mathcal S_j$.}
\KwOut{Updated records with $0\le r_{H_j}^{(j)}\le q_jb\mathbf d_{H_j}/(2^{12}N)$.}
Compute $\widehat h_{H_j}\leftarrow M_{H_jH_j}^{-1}r_{H_j}^{(j)}$ by an SDDM solve, component by component if necessary\;
Extend $\widehat h$ by zero and set $x\leftarrow x+\widehat h$\;
\For{$1\le t\le j$}{
  $y^{(t)}_v\leftarrow y^{(t)}_v+\widehat h_v$ for $v\in H_j$\;
}
\For{$0\le t\le j$}{
  $r^{(t)}_v\leftarrow r^{(t)}_v-(M_{H_jH_j}\widehat h_{H_j})_v$ for $v\in H_j$\;
  \For{$u\in (N(H_j)\setminus H_j)\cap H_{t-1}$}{
    $r^{(t)}_u\leftarrow r^{(t)}_u+(1-\alpha)\sum_{v\in N(u)\cap H_j}\widehat h_v$\;
  }
}
Update queue and density keys at touched rows; keep active sets, references, and budgets unchanged\;
\Return $\mathcal S_j$\;
\end{algorithm}

\Cref{alg:query-push,alg:full-solve-interface} update all affected
open records. Adjacency lists are stored, and queue or density keys
are updated only at touched rows. Opening a child may scan the parent
region once; later arrivals are registered when touched. Discarding
an interrupted call is lazy, with record creation and reclamation
charged to the operation that created the records. Completion of the
\emph{root} region $H_0$ takes priority and returns the exact answer
by \Cref{lem:potential-completion}.

\begin{algorithm}[!htbp]
\caption{\textnormal{\textsc{LeafNormalize}}: normalize residual at the leaf}
\label{alg:technical-leaf}
\small
\KwIn{$\ell,b,\theta$ and shared state $\mathcal S_\ell$, with an open leaf episode
  $Q_\ell=\theta b\ge q_\ell b$.}
\KwOut{Updated $\mathcal S_\ell$, with the same open episodes.}
\lIf{some open $i\le\ell$ has $W_i\ge Q_i$}{\Return $\mathcal S_\ell$}
Set $C\leftarrow\varnothing$ and initialize the queue
  $\mathcal D=\{v\in H_\ell:r_v\ge3\theta d_v\}$\;
\While{true}{
  $\mathcal S_\ell\leftarrow\textnormal{\textsc{ExteriorPush}}(\ell,\mathcal S_\ell)$, updating the shared
    records and $\mathcal D$ at every touched row\;
  \lIf{some open $i\le\ell$ has $W_i\ge Q_i$}{\Return $\mathcal S_\ell$}
  \lIf{$\mathcal D=\varnothing$}{\Return $\mathcal S_\ell$}
  Insert every vertex of $\mathcal D\setminus C$ into $C$\;
  Set $g_C\leftarrow r_C-\theta\mathbf d_C$\;
  Compute $h_C\leftarrow M_{CC}^{-1}g_C$ by an SDDM solve \tcp*{retains the floor $\theta\mathbf d_C$}
  Apply $h_C$ to $x$ and all open records using the update in
    \Cref{alg:full-solve-interface}, with $C$ as the update domain\;
  Update the inactive queues and $\mathcal D$ at the touched rows\;
}
\end{algorithm}

Write $q=q_\ell$ and $r=r^{(\ell)}$.  One invocation of
\textnormal{\textsc{LeafNormalize}} processes the fixed episode budget $Q_\ell$
and density $\theta=Q_\ell/b\ge q$; it does not reset any open counter.
The global component-completion test of \Cref{lem:potential-completion} has
priority over the local returns below. The cache is discarded on every
local return, while the shared potential,
active sets, residuals, and budget counters are retained.

The returned $\mathcal S_\ell$ always denotes the current shared records. A return with no exhausted open budget is normal; the counters
themselves identify an interrupted return. No additional return flag
is needed. A cache update creates no direct budget charge;
any ensuing single-vertex push or activation is charged when that event occurs.
Each open ancestor is updated by its state identity, and its block
solves remain suspended while the leaf procedure is running.

\section{Proofs of the recursive guarantees}
\label{sec:technical-proofs}

\subsection{\texorpdfstring{Proof of \Cref{lem:recursive-root}}{Proof of normalization composition}}
\label{rec:return-proof}

The induction uses the two cases of \eqref{rec:parent-child}: a
child either activates a vertex and retains one floor, or leaves it
outside with a capped reserve. We first state the discovery and leaf
guarantees used in this argument. Their proofs follow in
\Cref{sec:technical-query,sec:technical-leaf}.

\begin{lemma}\label{lem:query-push}
On a normal return, \textnormal{\textsc{ExteriorPush}} satisfies
$\vol(H_j)\le2N$ and
\Cref{condition:normalize-residual}(2), \Cref{condition:normalize-solution}.
The total adjacency-scan cost of level-$j$ pushes during one
level-$j$ call of \textnormal{\textsc{Normalize}} is $O(N)$.
\end{lemma}

\begin{lemma}\label{lem:recursive-leaf}
A level-$\ell$ call of \textnormal{\textsc{Normalize}} of scale $N$
satisfies $R_\ell(N)\le P(N^2/b+Nb)$.
A normal return satisfies
\Cref{condition:normalize-residual,condition:normalize-solution}.
\end{lemma}

Only the return assertion of \Cref{lem:recursive-leaf} is needed
below. Its proof for one episode does not use the episode count;
that count is used later to obtain the running time.

\begin{proof}[Proof of \Cref{lem:recursive-root}]
We use downward induction on $j$. If the initial drain leaves
$H_j=\varnothing$, the active condition is vacuous and the other
conditions follow from \Cref{lem:query-push}. We may therefore
consider a normal return after an episode.

For $j=\ell$, the claim follows from \Cref{lem:recursive-leaf}.
Assume it holds at level $j+1$. A successful child is called with
$q_{j+1}=\theta$. For each $v\in H_j$, there are two cases:
\begin{itemize}[leftmargin=*,itemsep=3pt]
\item If $v\in H_{j+1}$, then
$r_v^{(j)}=\theta d_v+r_v^{(j+1)}
 \le(1+\kappa_{j+1})\theta d_v=\kappa_j\theta d_v$.
\item If $v\in H_j\setminus H_{j+1}$, then
$r_v^{(j)}=4d_vy_v^{(j+1)}+r_v^{(j+1)}
 \le(4\cdot3/8+1)\theta d_v
 =5\theta d_v/2\le\kappa_j\theta d_v$.
\end{itemize}
These cases include vertices discovered during the child call.
The parent returns normally only when $\theta=q_j$, proving
\Cref{condition:normalize-residual}(1). The child's return drains
all exterior queues at levels $t\le j+1$, hence in particular those
at $t\le j$, proving part (2). The subsequent parent drain makes
no update because those queues are already empty.

Finally, an inactive level-$j$ coordinate changes only through
level-$j$ exterior pushes. Each such push caps it at $3q_j/8$
and activates it when the cap is reached. Joint updates are
nonnegative and supported on active vertices. Thus the two bounds
in \Cref{condition:normalize-solution} are preserved throughout,
which completes the induction.
\end{proof}

\subsection{\texorpdfstring{Proof of \Cref{lem:query-push}}{Proof of the exterior-push guarantees}}
\label{sec:technical-query}

We first justify the budget account. Write
$\mu_t=\sum_{v\in H_t}r_v^{(t)}$, and use the superscript
$\mathrm{start}$ for the drained snapshot at episode entry.
The thresholds of open calls satisfy $q_0\le\cdots\le q_t$,
since each child is opened at $\theta\ge q_{t-1}$.

\begin{lemma}\label{lem:mu_t-W_t}
It holds that $\mu_t\le\mu_t^{\mathrm{start}}+W_t$.
\end{lemma}

\begin{proof}
Fix the current endpoint and let $F=H_t$ at that endpoint.
Track its vertices from episode entry, writing $\Delta x\ge0$
for all intervening increments to $x$. Vertices already in
$H_t^{\mathrm{start}}$ contribute
$\mu_t^{\mathrm{start}}$ initially. For a vertex
$v\in F\setminus H_t^{\mathrm{start}}$, let $k$ be its first
inactive level at entry: $v\in H_{k-1}\setminus H_k$, with
$0\le k\le t$. This includes vertices outside $H_0$.
The drained snapshot gives $r_v^{(k),\mathrm{start}}\le q_kd_v$
and $y_v^{(k),\mathrm{start}}\le3q_k/8$.

Telescoping \eqref{rec:parent-child} and the root identity, the
residual $(\mathbf e_s-Mx)_v$ initially consists of the floors at
levels below $k$, its level-$k$ reserve, and $r_v^{(k),\mathrm{start}}$.
At the endpoint it consists of all floors through level $t$ and
$r_v^{(t)}$. Thus its initial contribution, after allowing for
these retained floors, is
\[
 a_v=r_v^{(k),\mathrm{start}}
       +\left(4y_v^{(k),\mathrm{start}}-\sum_{i=k}^t q_i\right)d_v
 \le\left(\frac52q_k-\sum_{i=k}^tq_i\right)d_v
 \le\frac32q_td_v\le\frac{t+3}{2}q_td_v.
\]
This is covered by the birth charge in
line~\ref{line:qp-birth} of \Cref{alg:query-push}.
Summing the residual identities on the fixed set $F$ gives
\[
 \mu_t=\mu_t^{\mathrm{start}}
       +\sum_{v\in F\setminus H_t^{\mathrm{start}}}a_v
       -\one_F^{\mathsf T}M\Delta x.
\]

All joint updates during this episode are supported on $F$:
ancestor solves are suspended, and the current and descendant
active sets are contained in $F$. The columns of $M_{FF}$ have
nonnegative sums, so these updates, as well as pushes supported
on $F$, contribute at most zero to the last term. A push outside
$F$ has depth at most $t$ and contributes at most its outgoing
flow $(1-\alpha)d_vh$. Line~\ref{line:qp-flow} charges that flow
to $W_t$. Consequently,
\[
 \mu_t\le\mu_t^{\mathrm{start}}
   +\sum_{v\in F\setminus H_t^{\mathrm{start}}}\frac{t+3}{2}q_td_v
   +(1-\alpha)\sum_{\substack{\text{pushes in the episode}\\
                              \text{of depth at most }t}}d_vh
 =\mu_t^{\mathrm{start}}+W_t.
\]
The crossing update is included because its charges are recorded
before the budget is tested.
\end{proof}

Fix a level-$j$ call of scale $N$; the next two estimates apply
before ancestor interruption.

\begin{lemma}\label{lem:query-push-volume}
The active set satisfies $\vol(H_j)\le2N$.
\end{lemma}

\begin{proof}
For $j\ge1$, \eqref{rec:parent-child} gives
\[
 q_j\vol(H_j)\le\sum_{v\in H_j}r_v^{(j-1)}
 \le\mu_{j-1}
 \le\mu_{j-1}^{\mathrm{start}}+W_{j-1}
 <2Q_{j-1}=2q_jN.
\]
The third inequality is \Cref{lem:mu_t-W_t}. At the root,
summing the root identity \eqref{nr:state} gives
$q_0\vol(H_0)\le1$, exactly as in \Cref{nr:push-account}.
The root identity and nonnegativity are preserved by all descendant
updates, so the same argument applies here.
\end{proof}

\begin{lemma}\label{lem:query-push-scalar}
Let $S_j=\sum d_vh$ over all level-$j$ pushes in the current call.
Then $S_j\le3q_jN/4=O(q_jN)$.
\end{lemma}

\begin{proof}
A coordinate starts at relative potential zero when it becomes
available. Until activation, only its level-$j$ pushes increase it.
Their cumulative increment is $y_v^{(j)}$ while it is inactive
and exactly $3q_j/8$ once it activates. Therefore
\[
 S_j=\frac38q_j\vol(H_j)
        +\sum_{v\in H_{j-1}\setminus H_j}d_vy_v^{(j)}.
\]
For $j\ge1$, the two cases of \eqref{rec:parent-child} imply
$r_v^{(j-1)}\ge q_jd_v$ on $H_j$ and
$r_v^{(j-1)}\ge4d_vy_v^{(j)}$ outside it. Hence
\[
 S_j\le\frac38\sum_{v\in H_{j-1}}r_v^{(j-1)}
 =\frac38\mu_{j-1}\le\frac34q_jN.
\]
At the root, the same expression for $S_0$ and the summed root
identity give $S_0\le3/8$, as in \Cref{nr:push-account}.
\end{proof}

\begin{proof}[Proof of \Cref{lem:query-push}]
The push rule preserves nonnegative residuals: its own row loses
at most its available work residual, and neighboring rows only
gain residual. Ancestor rows contain the same work residual plus
nonnegative retained terms, so their smaller own-row subtraction
also preserves nonnegativity. The activation cap preserves the
relative-potential bounds. If no budget interrupts the drain,
its stopping condition gives all the required exterior bounds.
The volume estimate is \Cref{lem:query-push-volume}.

We count level-$j$ pushes in two cases. A push that activates its
vertex is paid for once by its degree, for a total of $O(N)$.
A push that does not activate uses
$h=r_v^{(j)}/(5d_v)>q_j/5$. Its scan cost satisfies
$d_v<(5/q_j)d_vh$, so \Cref{lem:query-push-scalar} bounds the
total of these costs by $O(N)$ as well.

These bounds apply before an ancestor budget crosses. If the
crossing push itself has depth $j\ge1$, its pre-update row
satisfies $r_v^{(j-1)}\ge r_v^{(j)}>q_jd_v$ and
$\mu_{j-1}<2q_jN$, giving $d_v<2N$. Its final scan therefore
also costs $O(N)$. A crossing push at an earlier level is charged
to that earlier call. No further level-$j$ solve or drain is
performed after this interruption. Updating open records and
queue keys adds only depth and logarithmic factors, absorbed by
$P$. In particular, the scan bound proves that a drain cannot
continue indefinitely.
\end{proof}

\subsection{\texorpdfstring{Proof of \Cref{lem:recursive-layer}}{Proof of the recursive work bound}}
\label{sec:technical-layer}

Fix a level-$j$ call of scale $N$ and block size $b$. We prove
the same three estimates as in the nonrecursive analysis: a
comparison between consecutive full solves, a condition for halving,
and a bound on the exceptions. All counts are for finite execution
prefixes before ancestor interruption or global completion;
termination will follow from the counts.

Under the idealized SDDM convention of \Cref{sec:preliminaries},
\textsc{FullSolve} costs $\widetilde O(\vol(H_j))$ and its exact
increment $M_{H_jH_j}^{-1}r_{H_j}^{(j)}$ is nonnegative.
The same observation applies to a leaf cache solve with right-hand
side $r_C^{(j)}-\theta\mathbf d_C\ge0$. We retain the stated
residual tolerances below; exact solves satisfy them with zero error.

The comparison needs a reference fixed even for vertices that enter
the parent domain later. For the proof only, extend the records to
all of $V$. If $v\in H_{t-1}\setminus H_t$ for some $t<j$, set
\[
 \xi_v^{(j)}=\xi_v^{(t)}+\frac38\sum_{i=t}^{j-1}q_i,\qquad
 y_v^{(j)}=y_v^{(t)}-\frac38\sum_{i=t}^{j-1}q_i,\qquad
 r_v^{(j)}=r_v^{(t)}+\frac12\sum_{i=t}^{j-1}q_i d_v.
\]
Thus $x=\xi^{(j)}+y^{(j)}$ still holds. When $v$ activates at
level $t$, its new reference at level $t+1$ is
$x_v=\xi_v^{(t)}+3q_t/8$, exactly its previously extended value.
Repeating this argument shows that the extended $\xi^{(j)}$ is
fixed for the whole call. No records outside the current domain
are stored by the algorithm. The state identity becomes
\begin{align}
 \mathbf e_s-M\xi^{(j)}-\left(\sum_{t<j}q_t\right)\mathbf d
       -My^{(j)}
 =q_j\mathbf d_{H_j}+4D_{V\setminus H_j}y^{(j)}+r^{(j)}.
 \label{eq:rj_comparison}
\end{align}

Subtracting the reference changes the source term. The following
bound controls its negative part by the floors retained at earlier
levels.

\begin{lemma}\label{lem:bound_mu_next}
Throughout a level-$j$ call,
\[
 \mathbf e_s-M\xi^{(j)}-\left(\sum_{t<j}q_t\right)\mathbf d
 \ge-\left(1+\frac{3\alpha}{8}\right)
       \left(\sum_{t<j}q_t\right)\mathbf d.
\]
\end{lemma}

\begin{proof}
We use induction on $j$. For $j=0$, $\xi^{(0)}=0$ and the
claim is $\mathbf e_s\ge0$. Assume it holds at level $j-1$.
Evaluate the following two cases at the entry snapshot of the
level-$j$ call. Since the references are fixed, the resulting
bound then holds throughout the call.

For $v\notin H_{j-1}$, the extension gives
$\xi_v^{(j)}-\xi_v^{(j-1)}=3q_{j-1}/8$.
For $v\in H_{j-1}$, initialization gives
$\xi_v^{(j)}-\xi_v^{(j-1)}=y_v^{(j-1)}\ge3q_{j-1}/8$.
Consequently, at an exterior vertex,
\[
 \bigl(M(\xi^{(j)}-\xi^{(j-1)})\bigr)_v
 \le\frac{3\alpha}{8}q_{j-1}d_v.
\]
By the induction hypothesis,
\begin{align*}
 \left(\mathbf e_s-M\xi^{(j)}-\sum_{t<j}q_t\mathbf d\right)_v
 &=\left(\mathbf e_s-M\xi^{(j-1)}-\sum_{t<j-1}q_t\mathbf d\right)_v
   -q_{j-1}d_v-\bigl(M(\xi^{(j)}-\xi^{(j-1)})\bigr)_v\\
 &\ge-\left(1+\frac{3\alpha}{8}\right)
          \left(\sum_{t<j}q_t\right)d_v.
\end{align*}

For $v\in H_{j-1}$, initialization gives $\xi_v^{(j)}=x_v$.
Outside this set, the extended $y_u^{(j)}\le0$, so
$\xi_u^{(j)}\ge x_u$. The nonpositive off-diagonal entries of
$M$ imply $(M\xi^{(j)})_v\le(Mx)_v$. Telescoping the residual
identities on $H_{j-1}$ therefore gives
\[
 \left(\mathbf e_s-M\xi^{(j)}-\sum_{t<j}q_t\mathbf d\right)_v
 \ge\left(\mathbf e_s-Mx-\sum_{t<j}q_t\mathbf d\right)_v
 =r_v^{(j-1)}\ge0.
\]
The two cases complete the induction.
\end{proof}

Now consider the drained snapshots immediately before two consecutive
level-$j$ full solves. Let their active sets be $H$ and $H\cup Z$,
and their work-residual masses be $\mu_j$ and
$\mu_j^{\mathrm{next}}$. The first solve computes $h$ supported
on $H$, with
$M_{HH}h=r_H^{(j)}-f$ and
$0\le f\le q_jb\mathbf d_H/(2^{12}N)$.
Write the updates to $x$ between the snapshots as
$x\longrightarrow x+h\longrightarrow x+h+u$.
Here $u\ge0$ includes \emph{all} subsequent updates, including
descendant joint solves as well as exterior pushes. Since the
extended reference is fixed,
$y^{(j),\mathrm{next}}=y^{(j)}+h+u$.

Let $\mathcal P_j$ be the pushes of depth at most $j$ between
these snapshots. For a push $p$, write $v_p$ for its vertex and
$\delta_p$ for its increment. Define
\[
 D=\{v\in H:h_v>y_v^{(j)}\},\qquad
 J_j=4\kappa_0q_j\vol(Z)
       +(1-\alpha)\sum_{p\in\mathcal P_j}d_{v_p}\delta_p.
\]
Thus $D$ contains the coordinates doubled by the first solve,
whereas $J_j$ accounts for new vertices and exterior input.

\begin{lemma}\label{lem:residual_next_bound}
With the preceding notation,
\[
 \mu_j^{\mathrm{next}}
 \le J_j+2\kappa_0q_j\vol(D)
       +2\sum_{v\in D}r_v^{(j)}+\|f\|_1.
\]
\end{lemma}

\begin{proof}
For $v\in H$, the retained floor is unchanged, so
\eqref{eq:rj_comparison} gives
$r_v^{(j),\mathrm{next}}=f_v-(Mu)_v$.
For $v\in Z$, the exterior reserve is replaced by an active floor,
giving
\[
 r_v^{(j),\mathrm{next}}
 =r_v^{(j)}+(4y_v^{(j)}-q_j)d_v
       +(1-\alpha)A_{vH}h-(Mu)_v.
\]
Set $F=H\cup Z$. Summing the two expressions yields
\[
 \mu_j^{\mathrm{next}}=\|f\|_1
 +\sum_{v\in Z}\bigl((4y_v^{(j)}-q_j)d_v+r_v^{(j)}
                          +(1-\alpha)A_{vH}h\bigr)
 -\one_F^{\mathsf T}Mu.
\]
We bound the later updates and the newly active rows separately.

\emph{Later updates.}
Since $u\ge0$ and $M_{FF}$ has nonnegative column sums,
\[
 -\one_F^{\mathsf T}Mu
 =-\one_F^{\mathsf T}M_{FF}u_F
       +(1-\alpha)\one_F^{\mathsf T}A_{F,V\setminus F}u_{V\setminus F}
 \le(1-\alpha)\sum_{p\in\mathcal P_j}d_{v_p}\delta_p.
\]
Indeed, every descendant joint update and every deeper push is
supported inside the growing $H_j\subseteq F$; ancestor joint
solves are suspended. Only the pushes in $\mathcal P_j$ can
contribute outside $F$, and each is bounded by its own outgoing flow.

\emph{Newly active rows.}
Set $w=(h-y_H^{(j)})_+$, supported on $D$, so that
$h\le y_H^{(j)}+w$. For $v\in Z$, \eqref{eq:rj_comparison} gives
\begin{align*}
 (1-\alpha)A_{vH}y_H^{(j)}
 &=5d_vy_v^{(j)}+r_v^{(j)}
   -\left(\mathbf e_s-M\xi^{(j)}-\sum_{t<j}q_t\mathbf d\right)_v
   -(1-\alpha)A_{v,V\setminus H}y_{V\setminus H}^{(j)}.
\end{align*}
At the first drained snapshot, actual exterior records satisfy
$y_v^{(t)}\le3q_t/8$ and $r_v^{(t)}\le q_td_v$.
The extended records therefore satisfy, for $v\notin H$,
\[
 -\frac38\sum_{t<j}q_t\le y_v^{(j)}\le\frac38q_j,
 \qquad r_v^{(j)}\le\frac{j+2}{2}q_jd_v.
\]
Combining these bounds with \Cref{lem:bound_mu_next}, we obtain
\begin{align*}
 (1-\alpha)A_{vH}y_H^{(j)}
 &\le\left(\frac{15}{8}+\frac{j+2}{2}\right)q_jd_v
   +\left(1+\frac{3\alpha}{8}+\frac{3(1-\alpha)}8\right)
       \left(\sum_{t<j}q_t\right)d_v\\
 &=\left(\frac{15}{8}+\frac{j+2}{2}\right)q_jd_v
       +\frac{11}{8}\left(\sum_{t<j}q_t\right)d_v.
\end{align*}
Using $h\le y_H^{(j)}+w$ and $q_t\le q_j$, it follows that
\begin{align*}
 (4y_v^{(j)}-q_j)d_v+r_v^{(j)}+(1-\alpha)A_{vH}h
 &\le\left[\left(j+\frac{35}{8}\right)q_j
                +\frac{11}{8}\sum_{t<j}q_t\right]d_v
       +(1-\alpha)A_{vD}w_D\\
 &\le4\kappa_0q_jd_v+(1-\alpha)A_{vD}w_D.
\end{align*}

It remains to bound the extra flow from $w$. On $H\setminus D$,
$h-y_H^{(j)}\le0$, so the nonpositive off-diagonal entries give
\begin{align*}
 M_{DD}w_D
 &\le\bigl(M_{HH}(h-y_H^{(j)})\bigr)_D\\
 &=2r_D^{(j)}-f_D+q_j\mathbf d_D
   -\left(\mathbf e_s-M\xi^{(j)}-\sum_{t<j}q_t\mathbf d\right)_D
   -(1-\alpha)A_{D,V\setminus H}y_{V\setminus H}^{(j)}\\
 &\le2r_D^{(j)}
       +\left(q_j+\frac{11}{8}\sum_{t<j}q_t\right)\mathbf d_D
 \le2r_D^{(j)}+2\kappa_0q_j\mathbf d_D.
\end{align*}
Summing over $D$ counts all flow leaving $D$, including that
entering $Z$, together with a nonnegative teleportation term.
Hence
\[
 (1-\alpha)\one_Z^{\mathsf T}A_{ZD}w_D
 \le\one_D^{\mathsf T}M_{DD}w_D
 \le2\kappa_0q_j\vol(D)+2\sum_{v\in D}r_v^{(j)}.
\]
Combining the bounds proves the claim.
\end{proof}

Successful normalization supplies the density bound needed to turn
this comparison into a contraction.

\begin{lemma}\label{lem:technical-full-halving}
Suppose the first full solve follows a successful episode of budget
$Q_j$, with $\theta=Q_j/b>q_j$ and $\mu_j\ge Q_j/2$.
If $J_j<\mu_j/4$ and $\vol(D)\le b/(64\kappa_0)$, then
$\mu_j^{\mathrm{next}}<\mu_j/2$.
\end{lemma}

\begin{proof}
At the first snapshot,
$\mu_j>q_jb/2$ and
$r_v^{(j)}\le\kappa_j\theta d_v\le2\kappa_j\mu_jd_v/b$
for $v\in H$. Therefore
\[
 2\kappa_0q_j\vol(D)+2\sum_{v\in D}r_v^{(j)}
 \le\frac{8\kappa_0\mu_j}{b}\vol(D)\le\frac{\mu_j}{8}.
\]
Also, $\vol(H)\le2N$ gives
$\|f\|_1\le q_jb/2048<\mu_j/1024$.
\Cref{lem:residual_next_bound} now yields
$\mu_j^{\mathrm{next}}<(1/4+1/8+1/1024)\mu_j<\mu_j/2$.
\end{proof}

\begin{lemma}\label{lem:technical-full-exceptions}
For two consecutive full solves in one level-$j$ call,
$\mu_j^{\mathrm{next}}\ge\mu_j/2$ implies at least one of:
\begin{enumerate}[label=(\arabic*),leftmargin=*,itemsep=3pt]
\item The first solve is the initial solve, or was requested after
an episode exhausted its budget.
\item The first solve follows a successful episode, and
$J_j\ge\mu_j/4$.
\item The first solve follows a successful episode, and
$\vol(D)>b/(64\kappa_0)$.
\end{enumerate}
The total number of these events is
$O(\kappa_0(N/b)\log(2N_0))$.
\end{lemma}

\begin{proof}
If (1) does not occur, the first solve follows a successful episode
with $\theta>q_j$ and $\mu_j\ge Q_j/2$. If neither (2) nor (3)
occurs, \Cref{lem:technical-full-halving} gives
$\mu_j^{\mathrm{next}}<\mu_j/2$. It remains to count the three
events separately.

\emph{Budget restarts.}
Write $W_j^{(i)}$ for the total charge in episode $i$. Its charges
have three sources: level-$j$ pushes, earlier-level pushes, and
newly active vertices. The first contributes $O(q_jN)$ by
\Cref{lem:query-push-scalar}. For $j\ge1$, all earlier-level
outflow is also charged to the one parent episode containing this
call, whose allowance is $Q_{j-1}=q_jN$; before it is exhausted,
that outflow is below $q_jN$. At the root there is no earlier
level, and its own outflow is $O(1)=O(q_0N_0)$.
Finally, each vertex enters $H_j$ at most once, so
\Cref{lem:query-push-volume} bounds all birth charges by
$O(\kappa_0q_jN)$. Therefore
\[
 \sum_i W_j^{(i)}=O(\kappa_0q_jN).
\]
Every exhausted episode that causes a local restart spends at
least $q_jb$. If $m_1$ is the number of those episodes, then
$m_1q_jb\le\sum_iW_j^{(i)}$, giving $m_1=O(\kappa_0N/b)$.
Adding the initial solve bounds (1). An episode interrupted by
an ancestor need not satisfy this charge bound after its crossing
update, but it causes no local restart and adds at most one
terminal episode.

\emph{Large new input.}
Index the disjoint intervals between consecutive full solves by
$i$. Their newly active sets $Z_i$ are disjoint, and their push
flows are bounded by the preceding accounts. Hence
\[
 \sum_i J_j^{(i)}
 \le4\kappa_0q_j\sum_i\vol(Z_i)+O(q_jN)
 =O(\kappa_0q_jN).
\]
In case (2), $\mu_j>q_jb/2$, so $J_j\ge\mu_j/4>q_jb/8$.
Thus (2) occurs $O(\kappa_0N/b)$ times.

\emph{Large doubling sets.}
Before global completion, the root set $H_0$ is a proper subset
of the source component. The root state and
\Cref{lem:potential-completion} give $x_v\le3N_0$.
An actual reference is a previously recorded nonnegative value
of $x_v$, so $0\le y_v^{(j)}\le x_v\le3N_0$ on $H_j$.
Each vertex enters $H_j$ at $y_v^{(j)}=3q_j/8$, and every
membership in $D$ more than doubles this coordinate. The reference
is fixed and all intervening increments are nonnegative. Thus,
if a vertex belongs to $D$ at $k$ full solves,
\[
 2^k\frac{3q_j}{8}\le3N_0.
\]
Since $q_j\ge q_0=1/N_0$, we have $k=O(\log(2N_0))$.
Writing $D_i$ for the doubling sets and $m_3$ for the number of
occurrences of (3), we obtain
\begin{align*}
 m_3\frac{b}{64\kappa_0}
 &\le\sum_i\vol(D_i)
   =\sum_{v\text{ activated in this call}}d_v\,\#\{i:v\in D_i\}\\
 &\le O(N\log(2N_0)).
\end{align*}
This bounds (3) and completes the proof. The logarithm uses the
global scale $N_0$; no upper bound in terms of the smaller local
scale is needed.
\end{proof}

\begin{proof}[Proof of \Cref{lem:recursive-layer}]
We first count operations in a finite prefix, without assuming
termination. Call a pair of consecutive full solves exceptional
if one of (1)--(3) in \Cref{lem:technical-full-exceptions} occurs.
Otherwise the active mass halves. For $j\ge1$, the parent account
gives $\mu_j\le\mu_{j-1}<2q_jN$ before ancestor interruption;
at the root, $\mu_0\le1=q_0N_0$. The first solve in a
nonexceptional pair follows a successful episode with
$\mu_j\ge Q_j/2$ and $Q_j>q_jb$. Consequently,
\[
 q_jb/2<\mu_j\le2q_jN.
\]
There can be only $O(\log(2N/b))$ consecutive halvings in this
range. Partitioning the full-solve sequence at the exceptional
pairs therefore gives at most $PN/b$ full solves, after enlarging
$P$ to absorb the depth and logarithmic factors.

Next consider consecutive successful, nonterminal episodes that
do not request a full solve. Their output mass $\mu_j^{\mathrm{out}}$ is below $Q_j/2$,
so the next budget satisfies
\[
 Q_j^{\mathrm{next}}
 =\max\{\mu_j^{\mathrm{out}},q_jb\}
 \le\max\{Q_j/2,q_jb\}.
\]
Since $q_jb\le Q_j\le2q_jN$, a run has length
$O(\log(2N/b))$ before reaching the floor $q_jb$, where success
returns. Such runs are separated by full-solve requests or budget
restarts. Including a possible final interrupted episode, we get
\[
 \#\{\text{episodes}\}
 \le O(\log(2N))
       \bigl(\#\{\text{full solves}\}
             +\#\{\text{budget restarts}\}+1\bigr)
 \le PN/b.
\]
Each episode makes at most one child call.

Each full solve and each episode-entry or child-initialization
scan costs $\widetilde O(N)$, because the current active set has
volume at most $2N$. Level-$j$ discovery costs $O(PN)$ in total
by \Cref{lem:query-push}. Ancestor pushes encountered during
these operations are charged to their producing calls, including
any crossing update. Thus for $j<\ell$,
\[
 R_j(N)\le P\left(\frac{N^2}{b}
                      +\frac Nb R_{j+1}(b)\right)+O(PN).
\]
The leaf cache count in \Cref{sec:technical-leaf} gives
$R_\ell(N)\le P(N^2/b+Nb)$. All logarithms in these bounds are
at most $O(\log(2N_0))$, so the same $P$ works at every level.

Finally, these uniform finite-prefix bounds imply termination by
induction from the leaf. A leaf episode is finite by the
single-episode argument in \Cref{sec:technical-leaf}, and the
episode count prevents infinitely many episodes. At an internal
level, each drain is finite by \Cref{lem:query-push}; each child
call terminates or is interrupted by induction. The same episode
count then makes the parent call finite. Global completion ends
all calls immediately and costs an additional
$\widetilde O(N_0)$, charged to the root and absorbed in its bound.
\end{proof}

\subsection{\texorpdfstring{Proof of \Cref{lem:recursive-leaf}}{Proof of the leaf guarantees}}
\label{sec:technical-leaf}

\begin{proof}
Write $q=q_\ell$ and $r=r^{(\ell)}$. We first analyze one leaf
episode, with fixed $Q_\ell=\theta b$ and $\theta\ge q$.
Its cache grows monotonically. A vertex enters with
$r_v\ge3\theta d_v$, and every cache solve leaves
$\theta\mathbf d_C\le r_C\le(\theta+q/4)\mathbf d_C$.
Between these solves, no push acts on a cached vertex, because
it is active at every open level. Its residual can only increase
through incoming flow, so the retained floor persists.

Before every permitted cache solve, the budget is unexhausted.
By \Cref{lem:mu_t-W_t},
\[
 \theta\vol(C)\le\sum_{v\in C}r_v\le\mu_\ell
 <2Q_\ell=2\theta b.
\]
Thus every solved cache has volume below $2b$. A crossing update
may break this bound, but then the loop returns before another
solve. If the high-density queue becomes empty instead, all
$r_v<3\theta d_v$ on $H_\ell$ and the preceding uninterrupted
drain has emptied every exterior queue. This proves the
single-episode output guarantee independently of any episode count.

We next count the \emph{individual cache solves inside}
\textsc{LeafNormalize}. For any finite prefix with $E$ entered
episodes, split them into two cases, as in \Cref{nr:cache-count}.
A solve admitting a new vertex is paid for by that admission.
The cache is monotone within an episode and has volume below $2b$,
so there are $O(bE)$ such solves.

A solve with no new admission has a preceding cache solve in
the same episode. After the preceding solve, every cached row
has density at most $\theta+q/4$. Reaching $3\theta$ again
requires at least $(2\theta-q/4)d_v\ge q$ new residual mass
at some cached vertex. Only exterior-push flow can supply it:
ancestor joint solves are suspended, and no other joint update
occurs between these two cache solves. The intervals between
successive cache solves are disjoint.

The relevant flow over the entire leaf call is $O(qN)$.
Its own pushes have that bound by \Cref{lem:query-push-scalar}.
For $\ell\ge1$, all earlier-level outflow is charged to the
containing parent episode and is below
$Q_{\ell-1}=qN$ before interruption. For $\ell=0$, the root
outflow is $O(1)=O(q_0N_0)$. Flow from a final crossing update
cannot cause a subsequent cache solve. Therefore the number of
solves with no new admission is $O(N)$, and the total is
$O(bE+N)$. In particular, this bound with $E=1$, together with
the finite drain bound, rules out an infinite loop in one episode.

We can now apply the finite-prefix episode count from
\Cref{sec:technical-layer}, which used only the output guarantee
just proved. It gives $E\le PN/b$ and at most $PN/b$ full solves.
The $O(bE+N)=O(PN)$ cache solves each cost $\widetilde O(b)$;
the full solves and episode scans each cost $\widetilde O(N)$.
Adding discovery yields
\[
 R_\ell(N)\le P\left(\frac{N^2}{b}+Nb\right).
\]

For a normal return after an episode, $\theta=q$, so the output
guarantee gives $r_v^{(\ell)}<3qd_v=\kappa_\ell q_\ell d_v$
on $H_\ell$, with all exterior queues drained. The cap and
activation rule give \Cref{condition:normalize-solution}.
If the initial drain returns with $H_\ell=\varnothing$, the
active condition is vacuous and the others follow directly from
\Cref{lem:query-push}. This proves all return assertions.
\end{proof}

\clearpage
\bibliographystyle{alpha}
\bibliography{ref}

\end{document}